\documentclass[a4paper,11pt]{article}

\usepackage[backend=biber, style=alphabetic, sorting=anyt, maxnames=99]{biblatex}
\usepackage{graphicx}
\usepackage{amsmath, amsthm, amssymb}
\usepackage{enumitem}
\usepackage{xcolor}
\usepackage{hyperref}
\usepackage{macros}
\usepackage{scalerel}
\usepackage[left=1in, right=1in, top=1in, bottom=1in]{geometry}
\usepackage{authblk}
\usepackage{cleveref}

\renewbibmacro{in:}{}
\newcommand{\alg}[1]{\mathcal{#1}}

\numberwithin{equation}{section}

\title{The category of anyon sectors for non-abelian quantum double models}
\author[1]{Alex Bols}
\author[2]{Mahdie Hamdan}
\author[2]{Pieter Naaijkens}
\author[3]{Siddharth Vadnerkar}
\affil[1]{Institute for Theoretical Physics, ETH Z{\"u}rich, Switzerland}
\affil[2]{School of Mathematics, Cardiff University, United Kingdom}
\affil[3]{Department of Physics, University of California, Davis, CA, USA}

\date{\today}

\begin{document}

\maketitle

\begin{abstract}
We study Kitaev's quantum double model for arbitrary finite gauge group in infinite volume, using an operator-algebraic approach.
The quantum double model hosts anyonic excitations which can be identified with equivalence classes of `localized and transportable endomorphisms', which produce anyonic excitations from the ground state. 
Following the Doplicher--Haag--Roberts (DHR) sector theory from AQFT, we organize these endomorphisms into a braided monoidal category capturing the fusion and braiding properties of the anyons. We show that this category is equivalent to $\Rep_f \caD(G)$, the representation category of the quantum double of $G$. This establishes for the first time the full DHR structure for a class of 2d quantum lattice models with non-abelian anyons.
\end{abstract}

\setcounter{tocdepth}{2}
\tableofcontents


\section{Introduction}
Kitaev's quantum double model~\cite{kitaev2003fault} is the prototypical example of a topologically ordered quantum spin system with long-range entanglement (see~\cite{MR3929747} for an introduction).
Such models host quasi-particle excitations with non-trivial braid statistics called anyons.
The physical properties of such anyons (such as their behavior under exchange or fusion) can be described algebraically by braided (and often even modular) tensor categories~\cite{Kitaev2006,Wang2010}.
In this paper we show that for the quantum double model for a finite gauge group $G$, defined on the plane, this braided  tensor category can be recovered from the unique frustration-free ground state of the model (under some mild technical assumption), and is given by $\Rep_f \caD(G)$, the category of finite dimensional unitary representations of the quantum double algebra of $G$. 

Our approach is motivated by the Doplicher--Haag--Roberts (DHR) theory of superselection sectors (see~\cite{HaagLQP} for an overview).
Mathematically, we can identify the anyons with certain equivalence classes of irreducible representations of the (quasi-local) observable algebra $\alg{A}$.

The relevant representations are those whose vector states approximately agree with the model's ground state on observables supported far away from some fixed point (which we can take as the origin), and whose support does not encircle this point.
The latter condition is to exclude observables corresponding to braiding other anyons around the fixed point, which are able to distinguish non-trivial anyon states from states in the ground state sector.

This intuition is conveniently captured by the \emph{superselection criterion}. Namely, a representation $\pi$
 satisfies the superselection criterion if
\begin{equation}
    \label{eq:sselect}
    \pi | \alg{A}_{\Lambda^c} \cong \pi_0 | \alg{A}_{\Lambda^c},
\end{equation}
where $\Lambda$ is any cone (a notion which we will make more precise later) and $\Lambda^c$ is its complement,  $\pi_0$ is the GNS representation of the (unique) frustration free ground state of the quantum double model, and $\alg{A}_{\Lambda^c}$ is the ${\rm C}^*$-algebra generated by all local observables localized in $\Lambda^c$.
That is, we consider representations that, outside \emph{any} cone, are unitarily equivalent to the ground state representation.
A superselection sector (or simply anyon sector) is an equivalence class of such representations.

The key insight of Doplicher, Haag and Roberts is that the superselection sectors are naturally endowed with a monoidal product (`fusion') and a symmetry describing the exchange of bosonic/fermionic sectors. This was later extended to describe braiding statistics~\cite{frs1,frs2}, yielding a braided monoidal category.
These categories precisely capture the physical properties of anyon sectors, including their braiding and fusion rules.
The essential technical step is that, using a technical property called Haag duality, one can pass from representations to endomorphisms of the quasi-local algebra which are localized (i.e., they act non-trivially only in the localization region) and transportable (the localization region can be moved around with unitaries).
See~\cite{HalvorsonMueger} for an overview of this construction in the language of $\rm{C}^*$-tensor categories.
This theory was initially developed in the context of relativistic quantum field theories. 
The construction has later been adapted to quantum spin systems, see e.g.~\cite{Naaijkens2011,Fiedler2014,Ogata2021}.
For a recent completely axiomatic approach towards anyon sector theory, see~\cite{Bhardwaj2024}.

In this paper we study the anyon sector theory, including fusion and braiding rules, of the quantum double model for arbitrary finite gauge group $G$~\cite{kitaev2003fault}, extending previous results obtained for abelian $G$~\cite{Fiedler2014}. 
In particular, our main result can be paraphrased as follows:

\begin{theorem*}[Informal]
Let $\pi_0$ be the GNS representation of the frustration free ground state of the quantum double model for a finite group $G$ defined on the plane and assume that it satisfies Haag duality.
Then the category of representations satisfying~\eqref{eq:sselect} is braided monoidally equivalent to $\Rep_f \caD(G)$, the category of finite dimensional unitary representations of the quantum double algebra $\caD(G)$.
\end{theorem*}

We will give a precise statement of our main result (including our assumptions) later when we have introduced the necessary terminology, but remark that Haag duality for cones is a technical property that holds for the abelian quantum double model~\cite{Fiedler2014}, and one can still construct a category of anyon sectors without it (or with a weaker version thereof).
A proof of Haag duality for a large class of models has recently been announced~\cite{HaagDuality}.
See Remark~\ref{rem:Haag 2} below for more details.
We also note that since $\Rep_f \caD(G)$ is a unitary modular tensor category, the category of anyon sectors is as well.

As mentioned earlier, our assumptions imply that there is a braided $\rm{C}^*$-category of superselection sectors~\cite{Fiedler2014,Ogata2021}.
Our main contribution in this paper is to construct this category explicitly for the quantum double model for all finite groups $G$.
The main idea is as follows.
For each irreducible represention of $\caD(G)$, examples of representations $\pi$ satisfying the superselection criterion~\eqref{eq:sselect} were constructed in~\cite{Naaijkens2015}.
It was then shown in~\cite{bols2023classificationanyonsectorskitaevs} that these representations are irreducible, and in fact form a complete set of representatives of irreducible representations satisfying~\eqref{eq:sselect}.
These irreducible anyon sectors correspond to the simple objects (i.e., the anyon types) in our category.
Because we have a concrete description of the simple objects in our category, it is possible to explicitly implement the braiding and fusion operations defined abstractly in~\cite{Ogata2021}, and calculate those explicitly.
We then show that the category we constructed is indeed equivalent to the one defined abstractly in~\cite{Ogata2021}.

The key difference between the present work and the abelian case studied in~\cite{Naaijkens2011,Fiedler2014} is the use of \emph{amplimorphisms}, i.e. $*$-homomorphisms $\chi: \alg{A} \to M_n(\alg{A})$, instead of endomorphisms.\footnote{For technical reasons we will in fact need to consider amplimorphisms of some slightly bigger algebra $\alg{B}$ containing $\alg{A}$.}
This can be understood as follows: recall that in the quantum double models, we can define `ribbon operators' which create a pair of excitations from the ground state.
To obtain single-anyon states, one sends one of the excitations off to infinity.
For each irrep of $\caD(G)$, there is a corresponding multiplet of ribbon operators, transforming according to the irrep, with the total number of operators in the multiplet given by the dimension of the irrep. 
Hence for non-abelian representations, one has more than one ribbon operator, which combine naturally into an amplimorphism.

Although it is possible to pass from amplimorphisms to the endomorphisms used in~\cite{Fiedler2014,Ogata2021}, as we shall see later, doing so requires making some choices, and one loses the explicit description of the map.
Hence to identify the full superselection theory, we work mainly in the amplimorphism picture.
In particular, we show that the amplimorphisms constructed in~\cite{Naaijkens2015} can be endowed with a tensor product and a braiding, analogous to the tensor product and braiding of endomorphisms in the DHR theory.
More precisely, we construct a braided ${\rm C^*}$-tensor category $\Amp$ of localized and transportable amplimorphisms, which includes as objects the amplimorphisms constructed in~\cite{Naaijkens2015}.
We then consider the full subcategory $\Amp_f$ of $\Amp$ whose objects $\chi$ have finite dimensional $\Hom(\chi,\chi)$.
This category can be shown to be semi-simple and closed under the monoidal product on $\Amp$, and we study the fusion rules (how tensor products decompose into irreducible objects) and the braiding. The result is that the category $\Amp_f$ of such amplimorphisms is equivalent to $\Rep_f \caD(G)$ as braided tensor categories.
Using the classification result of anyon sectors in this model obtained by two of the authors~\cite{bols2023classificationanyonsectorskitaevs}, it then follows that the list of constructed anyon sectors is a complete list of representatives of irreducible anyon sectors. This then completes the classification. 

A similar approach using amplimorphisms was taken in~\cite{SzlachanyiV93,NillS97} to analyze topological defects of certain 1D quantum spin systems.
In their setting the anyon sectors are localized in finite intervals, with the corresponding algebra of observables localized in that region being finite dimensional.
This necessitated the use of amplimorphisms instead of endomorphisms.
In our case localization is in infinite cone regions, and the situation is different.
In particular, the unitary operators that can move the localization regions around no longer live in the quasi-local algebra $\alg{A}$.
From a technical point of view this means that we cannot restrict to a purely $\rm{C}^*$-algebraic approach with operators in the quasi-local algebra (or suitable amplifications) only, but have to consider von Neumann algebras as well, in particular the cone algebras $\pi_0(\cstar[\Lambda])''$.\footnote{This is already true for the abelian case, it is not specific to the non-abelian model.}
These cone algebras are ``big enough'' in the sense that they are properly infinite~\cite{Naaijkens2012,Fiedler2014,Tomba2023}.
This allows us to directly relate the localized and transportable amplimorphisms to localized and transportable endomorphisms of some suitably defined auxiliary algebra, making the connection with the usual DHR theory in terms of endomorphisms.

The paper is outlined as follows. In Section \ref{sec:setup and main result} we define the quantum double model and the associated categories of localized and tranportable amplimorphisms $\Amp$ and endomorphisms $\DHR$, as well as their `finite' versions $\Amp_f$ and $\DHR_f$.
We then state our main theorem, namely that the categories $\Amp_f$ and $\DHR_f$ are braided $\rm C^*$-tensor categories, equivalent to $\Rep_f \caD(G)$.
Section \ref{sec:braided tensor structure} is devoted to spelling out the braided $\rm C^*$-tensor structure of $\Amp$ and $\DHR$.
These two categories are then shown to be equivalent in Section \ref{sec:equivalence of Amp and DHR}.
Explicit localized and tranportable amplimorphisms corresonding to representations of $\caD(G)$ are constructed in Section \ref{sec:amplimorphism from ribbon operators} by taking limits of `ribbon multiplets'. These explicit amplimorphisms are organized into full subcategories $\Amp_{\rho}$ of $\Amp$ for a fixed half-infinite ribbon $\rho$, which are later shown to be equivalent to $\Amp_f$. This section also establishes the key properties of these ribbon multiplets that underlie the fusion and braiding structure of $\Amp_f$.
In Section \ref{sec:simples of Amp} we rephrase the main result of \cite{bols2023classificationanyonsectorskitaevs}, namely that the amplimorphisms corresponding to irreducible representations of $\caD(G)$ constructed in the previous section exhaust all simple objects of $\Amp$. Together with semi-simplicity of $\Amp_f$, this implies that the $\Amp_{\rho}$ are full and faithful subcategories of $\Amp_f$.
Finally, Section \ref{sec:proof of main THM} proves the main theorem. The appendices collect well-known facts about ribbon operators and some technical results related to taking their limits.

\textbf{Acknowledgements:}
We would like to thank Corey Jones, Boris Kj\ae r and David Penneys for helpful discussions.
MH was supported by EPSRC Doctoral Training Programme grant EP/T517951/1.
SV was funded by NSF grant number DMS-2108390.

\textbf{Copyright statement:} For the purpose of open access, the authors have applied a CC BY public copyright licence to any Author Accepted Manuscript version arising.

\textbf{Data availability:} We do not analyse or generate any datasets, because our work is entirely within a theoretical and mathematical approach.

\textbf{Conflict of interests:} The authors have no competing interests to declare that are relevant to the content of this article.

\section{Setup and main result} \label{sec:setup and main result}

\subsection{The quantum double model and its ground state}

\begin{figure}[t]
    \centering
    \includegraphics[width=0.5\linewidth]{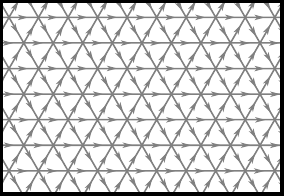}
    \caption{Snapshot of $\Gamma^E$. The edges are all oriented toward the right.}
    \label{fig:lattice snapshot}
\end{figure}

We first recall the definition of the quantum double model~\cite{kitaev2003fault} and introduce our notation.
Throughout the paper, we fix a finite group $G$.
Let $\Gamma$ be the triangular lattice in $\R^2$ and denote by $\Gamma^E$ the collection of oriented edges of $\Gamma$ which are oriented towards the right, see Figure~\ref{fig:lattice snapshot}.\footnote{We use the triangular lattice for simplicity, and to work in the same setting as~\cite{bols2023classificationanyonsectorskitaevs}, but believe the results hold for more general planar graphs as well.}
Denote by $\Gamma^V, \Gamma^F$ the set of vertices and faces of $\Gamma$ respectively. 
To each edge $e \in \Gamma^E$ we associate a degree of freedom $\caH_e \simeq \C[G]$ with basis $\{ |g\rangle_e \, : \, g \in G \}$ and corresponding algebra $\caA_e = \End(\caH_e) \cong M_{|G|}(\mathbb{C})$. We define in the usual way local algebras of observables $\caA^{\loc}_{X}$ supported on any $X \subset \Gamma^E$ and their norm closures $\caA_X := \overline{ \caA_X^{\loc}}^{\norm{\cdot}}$. We write $\caA = \caA_{\Gamma^E}$ and $\caA^{\loc} = \caA_{\Gamma^E}^{\loc}$.

The quantum double Hamiltonian is the commuting projector Hamiltonian given by the following formal sum
\begin{equation}\label{eq:Hamiltonian}
	H = \sum_{v \in \Gamma^V} \, (\1 - A_v) \, + \, \sum_{f \in \Gamma^F} \, ( \1 - B_f),
\end{equation}
where $A_v, B_f \in \cstar$ are the well-known \emph{star} and \emph{plaquette} operators of the quantum double model, which are mutually commuting projectors.
See Section \ref{subsubsec:gauge transformations and flux projections} in the appendix for precise definitions.

We say a state $\omega : \caA \rightarrow \C$ is a frustration free ground state of $H$ if
\begin{equation}
	\omega(A_v) = \omega(B_f) = 1
\end{equation}
for all $v \in \Gamma^V$ and all $f \in \Gamma^F$.
It is straightforward to verify that such a state $\omega$ indeed is a ground state for the dynamics generated by the Hamiltonian~\eqref{eq:Hamiltonian}.

The following theorem is proven in various sources~\cite{Fiedler2014, Cui2020kitaevsquantum, Tomba2023, bols2023classificationanyonsectorskitaevs}.
\begin{theorem}
	The quantum double Hamiltonian $H$ has a unique frustration free ground state which we denote by $\omega_0$. The uniqueness implies in particular that $\omega_0$ is pure.
\end{theorem}
We denote by $(\pi_0, \caH_0, \Omega)$ the GNS triple of the unique frustration free ground state $\omega_0$.
Note that $\pi_0$ is an irreducible representation since $\omega_0$ is pure.

\subsection{Cone algebras, Haag duality, and the allowed algebra}
\label{sec:cone_algebra}

The open cone with apex at $a \in \R^2$, axis $\hat v \in \R^2$, where $\hat v$ is a unit vector, and opening angle $\theta \in (0, 2\pi)$ is the subset of $\R^2$ given by
\begin{equation*}
	\Lambda_{a, \hat v, \theta} := \{ x \in \R^2 \, : \, (x - a) \cdot \hat v < \norm{x - a}_2 \cos(\theta/2)  \}.
\end{equation*}
We similarly define closed cones and call any subset of $\R^2$ that is either an open or a closed cone a \emph{cone}, so that the complement $\Lambda^c$ of any cone $\Lambda$ is again a cone. Note that a cone cannot be empty, nor can it equal the whole of $\R^2$.

For any $S \subset \R^2$ we denote by $\overline S$ the set of edges in $\Gamma^E$ whose midpoints lie in $S$. With slight abuse of notation we will simply write $S$ to mean the set of edges $\overline S$ unless otherwise stated.

To any cone $\Lambda$ we associate its \emph{cone algebra}
\begin{equation}
	\caR(\Lambda) := \pi_0(\caA_{\Lambda})'' \subset \caB(\caH_0).
\end{equation}
We remark that all these cone algebras are properly infinite factors~\cite{Naaijkens2012,Ogata2021}. We will moreover assume that Haag duality holds for cones.
\begin{assumption}[Haag duality for cones] \label{ass:Haag duality}
    For any cone $\Lambda$ we have
    $$\caR(\Lambda^c)' = \caR(\Lambda).$$
\end{assumption}

\begin{remark} \label{rem:Haag 2}
Haag duality for cones is proven in~\cite{Fiedler2014} in the case $G$ is an abelian group.
We believe the proof methods can be extended to the non-abelian case, however the analysis becomes considerably more technical since in the non-abelian case not all irreducible representations of the quantum double $\caD(G)$ are one-dimensional anymore.
In addition, a proof of Haag duality for a wide class of 2D quantum spin systems has been announced recently~\cite{HaagDuality}, including in particular for the non-abelian quantum double models considered here.

Finally, we comment on the role that Haag duality plays.
One can still construct the category of representations of superselection sectors, and show that the (equivalence classes of) irreducible representations are in one-to-one correspondence with the irreducible representations of $\caD(G)$~\cite{bols2023classificationanyonsectorskitaevs}.
By using this equivalence of categories the braided monoidal structure from $\Rep_f \caD(G)$ can be transported to the category of superselection sectors.
Haag duality is used to show that this in fact for example gives the natural braiding obtained from the Doplicher--Haag--Roberts approach.
That is, it has the correct physical interpretation.
Without Haag duality one can only do this for certain explicitly constructed representatives of each sector.\footnote{%
This is the category $\Amp_\rho$ that we will define later.
In this case, one can also explicitly construct the morphisms in the category as weak (or strong) operator limits of observables localized in some cone.
This gives enough control over the localization of these intertwiners, which requires Haag duality in general.
Using the explicit construction of the objects in the category, it can be directly checked that it is closed under the monoidal product of simple objects, and one can take finite direct sums.
However, this analysis only works for the amplimorphisms constructed explicitly, and does not extend to arbitrary ampli (or endo-)morphisms, even if they are in the same superselection sector.
}
For this reason, we prefer to assume (strict) Haag duality for cones to avoid making the analysis more technical than necessary.
\end{remark}

We fix a unit vector $\hat f \in \R^2$ and say a cone with axis $\hat v$ and opening angle $\theta$ is \emph{forbidden} if $\hat f \cdot \hat v < \cos(\theta /2)$. If a cone is not forbidden, then we say it is \emph{allowed}. The \emph{allowed algebra}
\begin{equation*}
	\caB = \caB_{\hat f} := \overline{\bigcup_{\Lambda \, \text{allowed}} \caR(\Lambda)}^{\norm{\cdot}} \subset \caB(\caH_0)
\end{equation*}
is the $\rm C^*$-algebra generated by the cone algebras of allowed cones.
Note that the set of allowed cones is a directed set for the inclusion relation.
Because we assume strict Haag duality, our algebra $\caB$ is the same as what is denoted by $\mathfrak{B}_{(\theta,\phi)}$ in~\cite[Eq. (2.5)]{Ogata2021} for suitable $(\theta,\phi)$.
If only approximate Haag duality holds, it can be replaced with the definition there.

Note that $\pi_0(\caA) \subset \caB$ as for any finite set $S \subset \Gamma^E$, we can find an allowed cone $\Lambda$ containing $S$. In addition, the allowed algebra will be seen to contain the intertwiners between the amplimorphisms we will consider. This will be crucial in defining the tensor product and the braiding.

It can be shown that the category of anyon sectors we define later does not depend on the choice of $\widehat{f}$.

\subsection{Categories of amplimorphisms and endomorphisms}

We largely follow the notation and terminology of~\cite{SzlachanyiV93}.
A *-homomorphism $\chi : \caB \rightarrow M_{n \times n}(\caB)$ is called an \emph{amplimorphism} of degree $n$.\footnote{One can take amplifications with infinite dimensional Hilbert spaces, but for our purposes it is enough to consider only the case where $n$ is finite.} We do not require such amplimorphisms to be unital. Given two amplimorphisms $\chi$ and $\chi'$ of degrees $n$ and $n'$ respectively, we let
\begin{equation} \label{eq:intertwiners of amplis defined}
    (\chi | \chi') := \{  T \in M_{n \times n'}(\caB(\caH_0)) \, : \, T \chi'(O) = \chi(O) T, \,\,\, O \in \caB, \,\,\, \chi(\1) T = T = T \chi'(\1) \}
\end{equation}
be the space of \emph{intertwiners} from $\chi'$ to $\chi$. The amplimorphisms $\chi$ and $\chi'$ are \emph{equivalent} if there is a partial isometry $U \in (\chi | \chi')$ such that $U^* U = \chi'(\1)$ and $U U^* = \chi(\1)$, in which case we write $\chi \sim \chi'$ and call $U$ an \emph{equivalence}.

An amplimorphism $\chi$ of degree $n$ is said to be \emph{localized} in a cone $\Lambda$ if for all $O \in \pi_0(\cstar[\Lambda^c])$, we have $\chi(O) = \chi(\1) (O \otimes \1_n)$. Such an amplimorphism is \emph{transportable} if for any cone $\Lambda'$ there is an amplimorphism $\chi'$ localized in $\Lambda'$ such that $\chi \sim \chi'$.

An amplimorphism $\chi$ is called \emph{finite} if the endomorphism space $(\chi | \chi)$ is finite dimensional. 
Note that $(\chi|\chi)$ is closed under taking adjoints. Hence if $\chi$ is finite, it follows that $(\chi|\chi)$ is isomorphic to a finite direct sum of full matrix algebras.

\begin{definition}
    We define $\Amp$ as the category whose objects are amplimorphisms that are localized in allowed cones, and are transportable.
    The morphisms between objects $\chi'$ and $\chi$ are given by $(\chi | \chi')$.
    The category $\Amp_f$ is the full subcategory of $\Amp$ whose objects are those amplimorphisms in $\Amp$ that are finite.
\end{definition}

In Section \ref{sec:braided tensor structure} we will show how the assumption of Haag duality allows us to endow $\Amp$ with the structure of a braided $\rm C^*$-tensor category. We will later see that the category $\Amp_f$ is closed under the monoidal product of $\Amp$ and therefore inherits the braided $\rm{C}^*$-tensor structure. The reduction to $\Amp_f$ is essential to establish equivalence with the category $\Rep_f \caD(G)$ of finite dimensional representations of the quantum double algebra. Indeed, $\Amp$ contains infinite direct sums, while $\Rep_f \caD(G)$ does not contain infinite direct sums by definition. We do not know if the infinite directs sums of objects of $\Amp_f$ exhaust all non-finite amplimorphisms of $\Amp$.

\begin{remark}
    In the algebraic description of anyons, it is commonly assumed that all anyons have a conjugate (see for example ~\cite[Sect. 6.3]{Wang2010}), meaning that each anyon type can fuse to the vacuum with some conjugate type.

    The assumption that an object in a $\rm C^*$-tensor category has a conjugate implies that it has a finite-dimensional endomorphism space~\cite[Lemma 3.2]{LongoRoberts97}. This is another way to see the necessity of restricting our attention to $\Amp_f$ if we want to show equivalence with $\Rep_f \caD(G)$. Indeed, all finite dimensional representations of $\caD(G)$ have conjugates. 
\end{remark}

\begin{definition}
    We denote by $\DHR$ the full subcategory of $\Amp$ whose objects are unital *-endomorphisms $\nu : \caB \rightarrow \caB$, \ie unital amplimorphisms of degree one. 
    Similarly, $\DHR_f$ is the full subcategory of $\DHR$ whose objects are finite endomorphisms.
\end{definition}
$\DHR$ is a braided $\rm C^*$-tensor subcategory of $\Amp$, see Section \ref{sec:braided tensor structure}. We show in Section \ref{sec:proof of main THM} that $\DHR_f$ is closed under the monoidal product of $\DHR$ and therefore inherits the braided $\rm C^*$-tensor structure.
The category $\DHR$ is equivalent to the category $\mathcal{O}_{\Lambda_0}$ defined in~\cite[Sect. 6]{Ogata2021}.
One can think of $\mathcal{O}_{\Lambda_0}$ as the subcategory of $\DHR$ restricted to endomorphisms localized in a specific cone $\Lambda_0$, however by the transportability requirement, one sees that this is equivalent to $\DHR$ (compare with Sect.~\ref{sec:proof of main THM} here).

\subsection{Main result}

We are now ready to give the main result of this paper, which states that the categories $\Amp_f$ and $\DHR_f$ introduced above are equivalent as braided $\rm C^*$-tensor categories to the category $\Rep_f \caD(G)$ of finite dimensional unitary representations of the quantum double $\caD(G)$ of the group $G$. See Appendix \ref{app:introduction to D(G)} for a brief review of $\caD(G)$ and its representation theory.

\begin{theorem} \label{thm:main result}
    If Haag duality for cones (Assumption \ref{ass:Haag duality}) holds, then the categories $\Amp_f$ and $\DHR_f$ are braided $\rm C^*$-tensor categories with monoidal structure and braiding as described in Section \ref{sec:braided tensor structure}. Moreover, both of these categories are then equivalent to $\Rep_f \caD(G)$ as braided $\rm C^*$-tensor categories.
\end{theorem}

Since $\Rep_f \caD(G)$ is a unitary modular tensor category (UMTC), it follows from this Theorem that $\Amp_f$ and $\DHR_f$ are also UMTCs. In particular, the anyon sectors are endowed with a duality which is inherited from the duality of finite dimensional representations of $\caD(G)$.


\section{Braided \texorpdfstring{$\rm C^*$}{}-tensor structure of \texorpdfstring{$\Amp$}{Amp} and \texorpdfstring{$\DHR$}{DHR}} \label{sec:braided tensor structure}

We spell out the $\rm C^*$-category structure of $\Amp, \Amp_f$, $\DHR$, and $\DHR_f$, as well as their finite direct sums and subobjects in Section~\ref{subsec:directsum_subojects}.

In Section \ref{subsec:braided monoidal} we use the assumption of Haag duality to endow $\Amp$ and $\DHR$ with braided ${\rm C}^*$-tensor structure. Most arguments in this section are straightforward adaptations of well-known constructions in the DHR superselection theory, see for example~\cite{SzlachanyiV93,NillS97,HalvorsonMueger,Ogata2021}. At this stage we do not know if the categories $\Amp_f$ and $\DHR_f$ are closed under the tensor product which we define for $\Amp$ and $\DHR$, a fact which will only be established in Proposition \ref{prop:equivalence of Amp_rho and Amp} and Lemma \ref{lem:DHR_f is braided monoidal} of Section \ref{sec:proof of main THM}.

\subsection{\texorpdfstring{$\rm C^*$}{C*}-structure, direct sums, and subobjects}
\label{subsec:directsum_subojects}

Let us first remark that the categories $\Amp$, $\Amp_f$, $\DHR$, and $\DHR_f$ are $\rm C^*$-categories (see~\cite{GhezLimaRoberts} or ~\cite[Definition 2.1.1]{neshveyev2013compact}). In this subsection we show that all these categories have finite direct sums and subobjects.

\subsubsection{Direct sums and subobjects of amplimorphisms} \label{subsec:direct sums and subobjects for Amp}

The direct sum of $\chi : \caB \rightarrow M_m(\caB)$ and $\psi : \caB \rightarrow M_n(\caB)$ is the amplimorphism $\chi \oplus \psi : \caB \rightarrow M_{m + n}(\caB)$ that maps $O \in \caB$ to the block diagonal matrix with blocks $\chi(O)$ and $\psi(O)$ with obvious projection and inclusion maps. If $\chi$ and $\psi$ are finite, then so is $\chi \oplus \psi$, so $\Amp_f$ is closed under this direct sum.

Before showing the existence of subobjects for $\Amp$ and $\Amp_f$, we state and prove two lemmas which will also be used to later to equip $\Amp$ and $\DHR$ with a tensor product.

\begin{lemma} \label{lem:localized amplimorphisms amplify locally}
    Let $\Lambda$ be an allowed cone. If $\chi$ is a $\Lambda$-localized amplimorphism of degree $n$, then $\chi(\caR(\Lambda)) \subset M_{n}( \caR(\Lambda) )$.
\end{lemma}

\begin{proof}
    If $O \in \caR(\Lambda^c)$ then the $\Lambda$-localization of $\chi$ implies that all components of $\chi(\1)$ commute with $O$, so $\chi(\I) \in M_n(\caR(\Lambda^c)') = M_n(\caR(\Lambda))$ by Haag duality.
    Now let $O \in \pi_0( \caA_{\Lambda^c})$ and $A \in \caR(\Lambda)$, then
    \begin{equation}
        \chi(O A) = \chi(O) \chi(A) = \chi(\1) (O \otimes \I_n) \chi(A) = (O \otimes \1_n) \chi(\I) \chi(A) = (O \otimes \I_n) \chi(A),
    \end{equation}
    but $\chi(O A) = \chi(A O)$ and by a similar computation we conclude that $(O \otimes \I_n) \chi(A) = \chi(A) (O \otimes \I_n)$. It follows that $\chi(A) \in M_{n}( \pi_0( \caA_{\Lambda^c})' ) = M_{n}( \pi_0( \caR(\Lambda^c)' ) = M_{n}( \caR(\Lambda) )$ by Haag duality.
\end{proof}

\begin{lemma} \label{lem:localized morphisms for AMP}
    If $\chi_1, \chi_2$ are localized transportable amplimorphisms of degrees $n_1$ and $n_2$, and localized on cones $\Lambda_1, \Lambda_2$ respectively, and $\Lambda$ is a cone that contains $\Lambda_1 \cup \Lambda_2$, then $(\chi_1 | \chi_2) \subset M_{n_1 \times n_2}( \caR(\Lambda) )$.
\end{lemma}

\begin{proof}
    If $T \in (\chi_1 | \chi_2)$ then for any $O \in \pi_0(\caA_{\Lambda^c})$ we have $\chi_1(O) T = T \chi_2(O)$. Since $O$ is supported outside of the cones $\Lambda_1, \Lambda_2$ on which $\chi_1$ and $\chi_2$ are localized, this implies $\chi_1(\1) (\pi_0(O) \otimes \I_{n_1}) T = T \chi_2(\1) (O \otimes \I_{n_2})$ for any $O \in \caA_{\Lambda^c}$. Using $\chi_1(\1) \in M_{n_1 \times n_2}( \caR(\Lambda_1) )$ (Lemma \ref{lem:localized amplimorphisms amplify locally}) and $\chi_1(\1) T = T = T \chi_2(\1)$ it follows that each component of $T$ belongs to $\pi_0(\caA_{\Lambda^c})' = \caR(\Lambda)$, where we used Haag duality. This proves the claim.
\end{proof}

We now establish the existence of subobjects.
Since at the moment we allow non-unital amplimorphisms, the construction is somewhat more elementary than the corresponding result for DHR endormorphisms (cf.~\cite[Lemma 5.8]{Ogata2021}).

\begin{proposition} \label{prop:existence of subobjects for Amp}
    Let $\chi \in \Amp$ and $p \in (\chi | \chi)$ an orthogonal projector. Then there are localized and transportable amplimorphisms $\chi_1, \chi_2 \in \Amp$ and partial isometries $v \in (\chi | \chi_1)$, $w \in (\chi | \chi_2)$ such that $v v^* = p, w w^* = \chi(\1) - p$ ands $v v^* + w w^* = \chi(\1)$. In particular, $\chi$ is isomorphic to $\chi_1 \oplus \chi_2$. If $\chi$ is finite, then so are $\chi_1$ and $\chi_2$.
\end{proposition}

\begin{proof}
    Consider the amplimorphism $\chi_1 : \caB \rightarrow M_n(\caB)$ given by $\chi_1(O) := p \chi(O) p$. By Lemma \ref{lem:localized morphisms for AMP} we have $p \in M_{n}( \caR(\Lambda) )$ where $n$ is the degree of $\chi$, so $\chi_{1}$ is localized on $\Lambda$. Moreover, $p \chi_1(O) = p \chi(O) p = \chi(O) p$ and $\chi(\1) p = p = p \chi_1(\1)$ which shows that $p \in (\chi | \chi_1)$.
    
    The amplimorphism $\chi_1$ is also transportable. Indeed, let $\Lambda'$ be some other cone. By transportability of $\chi$ there is an amplimorphism $\chi'$ of degree $n'$ localized on $\Lambda'$ and an equivalence $U \in (\chi|\chi')$. Consider the projection $q = U^* p U \in (\chi' | \chi') \subset M_{n'}( \caR(\Lambda') )$ and corresponding amplimorphism $\chi'_q(O) := q \chi'(O) q$ localized on $\Lambda'$. Then $$pU \chi'_q(O) = p U U^* p U \chi'(O) U^* p U = p \chi(\1) p \chi(O) p U = p \chi(O) p U = \chi_1(O) pU$$ and $pUU^*p = p \chi(\1) p = \chi_1(\1)$ while $$U^* p p U = U^* p \chi(\1) p U = q \chi'(\1) q = \chi'_q(\1),$$ so $pU$ is an equivalence of $\chi_1$ and $\chi'_q$.

    The same construction yields a localized transportable amplimorphism $\chi_{2}$ corresponding to the orthogonal projector $q = \chi(\1) - p \in (\chi | \chi)$. One easily checks that the claim of the proposition is satisfied with $v = p$ and $w = q$.

    Suppose $\chi_1$ were not finite, \ie $(\chi_1 | \chi_1)$ is infinite dimensional. Since $(\chi_1 | \chi_1)$ is isomorphic to $p (\chi | \chi) p$, this implies that $\chi$ is also not finite. With a similar argument for $\chi_2$, this shows that if $\chi$ is finite, then so are $\chi_1$ and $\chi_2$.
\end{proof}

\subsubsection{Direct sums and subobjects in \texorpdfstring{$\DHR$}{DHR}} \label{subsec:direct sums and subobjects for DHR}

The subcategory $\DHR$ is not closed under the direct sum described above, neither does the construction of subobjects stay in the $\DHR$ subcategory. However, $\DHR$ does have finite direct sums and subobjects, see \cite{Ogata2021}. The subcategory $\DHR_f$ is closed under these direct sums, and any subobject of a finite endomorphism must again be finite, so that $\DHR_f$ also has finite direct sums and subobjects.

\subsection{Braided \texorpdfstring{$\rm C^*$}{C*}-tensor structure of \texorpdfstring{$\Amp$}{Amp} and \texorpdfstring{$\DHR$}{DHR}} \label{subsec:braided monoidal}

Using the assumption of Haag duality for cones, we equip $\Amp$ and $\DHR$ with a monoidal product and a braiding, making them into braided $\rm{C}^*$-tensor categories (see Definition 2.1.1 of \cite{neshveyev2013compact}). At this point it is not clear that the tensor product of two finite amplimorphisms, as defined below, is again finite (and in fact one can construct examples of irreducible anyon sectors whose monoidal product decomposes into infinitely many irreducibles, see for example~\cite{Fredenhagen94}).
For this reason we can't yet equip $\Amp_f$ and $\DHR_f$ with the structure of braided $\rm{C}^*$-tensor categories. It will be shown in Proposition \ref{prop:equivalence of Amp_rho and Amp} and Lemma \ref{lem:DHR_f is braided monoidal} that $\Amp_f$ and $\DHR_f$ are in fact closed under the tensor product, and are therefore full braided $\rm C^*$-tensor subcategories of $\Amp$ and of $\DHR$ respectively.

\subsubsection{Monoidal structure}  \label{subsec:monoidal structure}

If $\chi : \caB \rightarrow M_n(\caB)$ is an amplimorphism of degree $n$ we denote by $\chi(O)^{i j}$ for $i, j = 1, \cdots, n$ the $\caB$-valued matrix components of $\chi(O)$. We endow $\Amp$ with a monoidal product $\times$ defined as follows. If $\chi_1$ and $\chi_2$ are amplimorphisms of degrees $n_1$ and $n_2$ respectively, then we define their tensor product $\chi_1 \times \chi_2 : \caB \rightarrow M_{n_1} \big( M_{n_2}(\caB) \big) \simeq M_{n_1 n_2}(\caB)$ to be the amplimorphism of degree $n_1 n_2$ with components
\begin{equation}
\label{eq:monoidal product}
	(\chi_1 \times \chi_2)^{u_1 u_2, v_1 v_2}(O) = \chi_1^{u_1 v_1} \big(  \chi_2^{u_2 v_2}(O)   \big) \quad \text{for all } \,  O \in \caB.
\end{equation}
Note that this is just $(\chi_1 \otimes \I_{n_2}) \circ \chi_2$ after identifying $\caB \otimes M_n(\mathbb{C})$ with $M_n(\caB)$.
For intertwiners $T \in (\chi | \chi')$ and $S \in (\psi | \psi')$ the tensor product $T \times S \in (\chi \times \psi | \chi' \times \psi')$ is defined by
\begin{equation} \label{eq:tensor of intertwiners}
	(T \times S)^{u_1 u_2, v_1 v_2} = \sum_{w_1, w_2} \chi^{u_1 w_1}( S^{u_2 w_2} ) T^{w_1 v_1} \delta^{w_2, v_2}
\end{equation}
which can also be written in matrix notation as $T \times S = \chi(S) (T \otimes I_{\psi'}) = (T \otimes I_{\psi}) \chi'(S)$.

The monoidal unit is the identity amplimorphism which is irreducible because $\caB'' = \caB(\caH_0)$ since $\pi_0$ is irreducible and $\pi_0(\alg{A}) \subset \alg{B}$.
Since the monoidal product is strict, it is trivially compatible with the ${\rm C}^*$-structure.
The subcategory $\DHR$ is closed under this monoidal product and contains the identity, it is therefore a monoidal subcategory of $\Amp$.

The monoidal product of objects is well defined thanks to Lemma \ref{lem:localized amplimorphisms amplify locally} and the monoidal product of intertwiners is well defined thanks to Lemma \ref{lem:localized morphisms for AMP}.
The monoidal product on $\DHR$ coincides with that defined in~\cite{Ogata2021} (see also the remarks around equations~(1.28)--(1.29) there).

\subsubsection{Braiding} \label{subsec:braiding}

It is well known that the category of localized endomorphisms for models in two spatial dimensions can be given a braiding \cite{frs1, frohlich1990braid, frohlich1988statistics}. Here we extend this to localized amplimorphisms.

The braiding on $\Amp$ is given by intertwiners $\ep(\chi, \psi) \in ( \psi \times \chi | \chi \times \psi)$ defined as follows. Since $\chi$ and $\psi$ are localized in allowed cones there is an allowed cone $\Lambda$ such that $\chi$ and $\psi$ are both localized in $\Lambda$. Let $\Lambda_L$ and $\Lambda_R$ be allowed cones `to the left and to the right' of $\Lambda$, \cf~Figure \ref{fig:braiding setup}. Let $\chi_R$ be a transportable amplimorphism localised in $\Lambda_R$ and fix an equivalence $U \in (\chi_R | \chi)$ with $U \in  M_m(\caR(\widetilde \Lambda_R))$ where $\widetilde \Lambda_R$ is an allowed cone that contains $\Lambda$ and $\Lambda_R$, but is disjoint from $\Lambda_L$. Similarly, pick a transportable amplimorphism $\psi_L$ localised in $\Lambda_L$ and a unitary $V \in (\psi_L | \psi)$ with $V \in M_n(\caR( \widetilde \Lambda_L))$. Such $\chi_R, U, \psi_L, V$ exist by transportability of $\chi$ and $\psi$. Now put
\begin{equation} \label{eq:braiding of amplimorphisms defined}
	\ep(\chi, \psi) := (V^* \times U^*) \cdot P_{12} \cdot (U \times V)
\end{equation}
where $P_{12} \in (\chi_R \times \psi_L | \psi_L \times \chi_R)$ is given by its components $P_{12}^{u_1 u_2, v_1 v_2} = \psi_L^{u_2 v_1}\left(\chi_R^{u_1 v_2}(\I)\right)$ (note the transposition of the indices compared to~\eqref{eq:monoidal product}).
That $P_{12}$ indeed is an intertwiner follows from a short calculation using that $\psi_L$ and $\chi_R$ are localized in disjoint cones, and hence $\psi_L^{u_1 u_2}(\chi_R^{v_1 v_2}(A)) = \chi_R^{v_1 v_2}(\psi_L^{u_1 u_2}(A))$ for all $A \in \caA$.
Alternatively,
\begin{equation}
    P_{12} = (\operatorname{id}_\caB \otimes P)((\psi_L \times \chi_R)(\I)), 
\end{equation}
where $P : M_n(\mathbb{C}) \otimes M_m(\mathbb{C}) \to M_m(\mathbb{C}) \otimes M_n(\mathbb{C})$ flips the tensor factors.
Using standard arguments, on can check that indeed $\ep(\chi, \psi) \in (\psi \times \chi | \chi \times \psi)$, that $\ep(\chi, \psi)$ is independent of the choices of $\chi_R, \psi_L, U, V$, and that $\ep$ is indeed a braiding for $\Amp$. See for example~\cite[Prop. 5.2]{SzlachanyiV93} for amplimorphisms, or \cite[Lemma 4.8]{Naaijkens2011},  \cite[Definition 4.10]{Ogata2021}, or \cite[Lemma 2.9]{bols2024double} for proofs of the analogous fact for the braiding of endomorphisms.\footnote{Note that in the case of approximate Haag duality (as in~\cite{Ogata2021}), one has to do some additional limiting procedure to define the braiding. This is because under the weaker localization properties, we do not necessarily have that $\rho \times \sigma = \sigma \times \rho$ if $\rho$ and $\sigma$ are approximately localized in disjoint cones.}
This braiding restricts to the $\rm C^*$-tensor subcategory $\DHR$, so $\DHR$ is a braided $\rm C^*$-tensor subcategory of $\Amp$.

\begin{figure}[!ht]
    \centering
    \includegraphics[width = 0.5\textwidth]{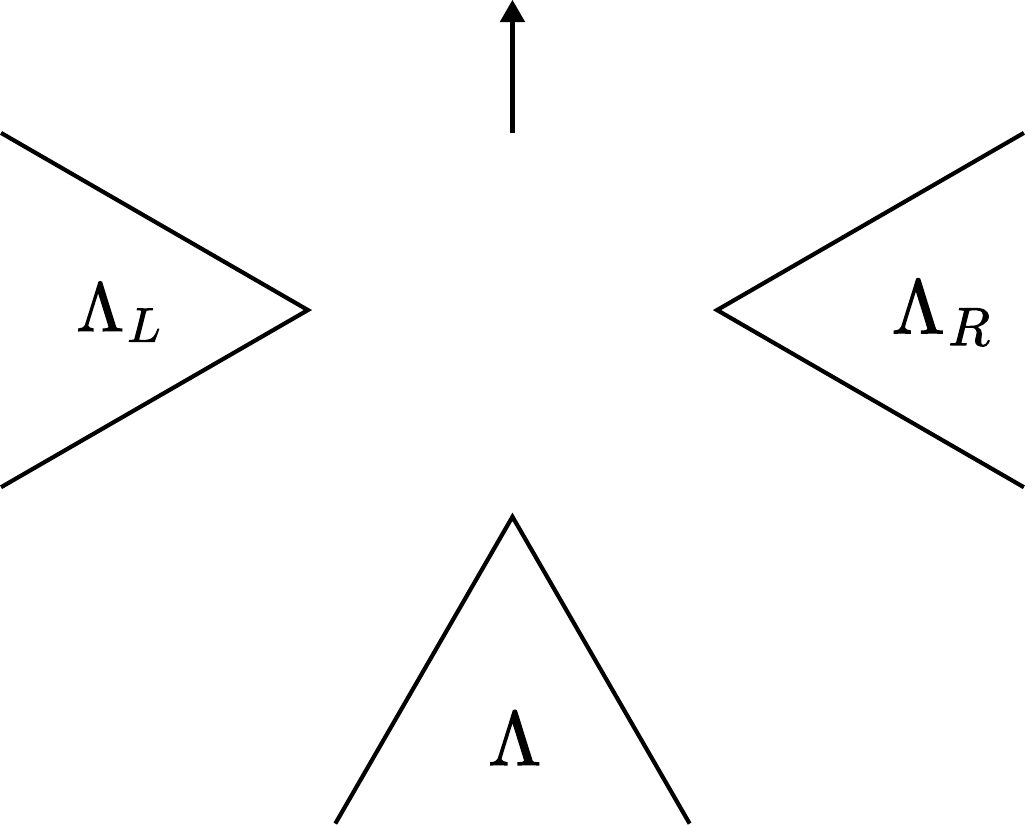}
    \caption{An example of the braiding setup. The arrow represents the forbidden direction.}
    \label{fig:braiding setup}
\end{figure}



\section{Equivalence of \texorpdfstring{$\Amp$}{Amp} and \texorpdfstring{$\DHR$}{DHR}}
\label{sec:equivalence of Amp and DHR}

\subsection{Reduction to unital amplimorphisms}

Our proof of the braided monoidal equivalence will rely on the fact that any amplimorphism of $\Amp$ is equivalent to a unital amplimorphism, a fact which we prove here. This fact will also be useful in Section \ref{sec:simples of Amp}, where the simple objects of $\Amp$ are characterized.

We say $\widetilde \Lambda$ is slightly larger than $\Lambda$, denoted $\Lambda \Subset \widetilde \Lambda$, if there exists another cone $\Lambda' \subset \widetilde \Lambda$ disjoint from $\Lambda$. 
That is, we can fit a cone in $\Lambda^c \cap \widetilde{\Lambda}$.
The following Lemma is proven in exactly the same way as \cite[Lemma 5.11]{Ogata2021}, and noting \cite[Corollary 6.3.5]{KadisonR97}.
We include it here for the convenience of the reader, as we will use this result repeatedly.

\begin{lemma} \label{lem:projectors are infinite}
    Let $\Lambda \Subset \widetilde \Lambda$ and let $p \in M_n( \caR(\Lambda) )$ be an orthogonal projector. Then $p$ is infinite as a projector in $M_n(\caR(\widetilde \Lambda))$, and is Murray-von Neumann equivalent to $\1_n$.
\end{lemma}
\begin{proof}
     By assumption, there is a cone $\Lambda' \subset \widetilde \Lambda$ that is disjoint from $\Lambda$. Since $\caR(\Lambda')$ is an infinite factor (see Sect.~\ref{sec:cone_algebra}), so is $M_n(\caR(\Lambda'))$ and we can apply the halving lemma~\cite[Lemma 6.3.3]{KadisonR97} to find an isometry $V \in M_n(\caR(\Lambda'))$ such that $V V^* < \1_n$. Note that $V$ and $V^*$ commute with $p$ since they have disjoint supports.
     The map $x \mapsto xp$ for $x \in M_n(\caR(\Lambda'))$ is a $*$-isomorphism from $M_n(\caR(\Lambda'))$ onto $M_n(\caR(\Lambda')p$ by~\cite[Prop. 5.5.5]{Kadison1983}.
     In particular, this implies that $V V^* p \neq p$, and hence is a proper subprojection of $p$.
     Put $\widetilde V = p V$, then
     \begin{equation}
         \widetilde V \widetilde V^* = p V V^* < p, \quad \widetilde V^* \widetilde V = p V^* V p = p.
     \end{equation}
     This shows that $p$ as a projection in $M_n(\caR(\widetilde \Lambda))$ is Murray von Neumann equivalent to its proper subprojection $p VV^*$ and thus $p$ is infinite in $M_n(\caR(\widetilde \Lambda))$. Murray-von Neumann equivalence to $\1_n$ now follows immediately from Corollary 6.3.5 of \cite{KadisonR97}.
\end{proof}

\begin{lemma} \label{lem:localised amplis are equiv to unital amplis}
        Let $\chi$ be an amplimorphism of degree $n$ localized in a cone $\Lambda$, and $\widetilde \Lambda$ be another cone such that $\Lambda \Subset \widetilde \Lambda$. Then there exists a unital amplimorphism localized on  $\widetilde \Lambda$ that is equivalent to $\chi$.
\end{lemma}

\begin{proof}
    By Lemma \ref{lem:localized amplimorphisms amplify locally} we have that the projector $\chi(\1)$ belongs to $M_n(\caR(\Lambda))$. By Lemma~\ref{lem:projectors are infinite}, it follows that $\chi(\1)$ is infinite as an element of $M_n(\caR(\widetilde\Lambda))$ and is Murray-von Neumann equivalent to $\1_n \in M_n(\caR(\widetilde{\Lambda}))$. Therefore there exists an isometry $V \in M_n(\caR(\widetilde\Lambda))$ such that $V V^* = \chi(\1)$ and $V^* V = \1_n$. 
    
    Let $\psi$ be given by $\psi(O) = V^* \chi(O) V$ for all $O \in \caB$, then $\psi(\1) = V^* \chi(\1) V = V^* V V^* V = \1_n$ so $\psi$ is indeed unital. In fact, we see that $V \in (\chi | \psi)$ is an equivalence.
    If $O \in \pi_0(\alg{A}_{\widetilde \Lambda^c})$ then 
    \begin{equation}
        \psi(O) = V^* \chi(O) V = V^* \chi(\1) (O \otimes \1_n) V = V^* (O \otimes \1_n) V = (O \otimes \1_n) V^* V = O \otimes \1_n,
    \end{equation}
    so $\psi$ is indeed localized on $\widetilde \Lambda$.
\end{proof}

If $\chi$ is in addition transportable, we can first transport to a smaller cone inside the localization region $\Lambda$, to make room for the `additional cone' needed in the proof.
The construction above does not affect transportability, so we immediately obtain the following corollary.

\begin{corollary} \label{cor:ampli is equivalent to unital ampli in the same cone}
    Any localized and transportable amplimorphism $\chi$ is equivalent to a \emph{unital} transportable amplimorphism $\chi'$ localized in the same cone.
\end{corollary}

\begin{proof}
    Let $\chi$ be localized in $\Lambda$. We have by transportability of $\chi$ that there exists an amplimorphism $\psi$ localized in a cone $\Lambda' \Subset \Lambda$ such that $\psi \sim \chi$. We have by Lemma \ref{lem:localised amplis are equiv to unital amplis} that there exists a unital amplimorphism $\chi'$ localized in $\Lambda$ such that $\chi' \sim \psi$, so we have $\chi' \sim \chi$. Transportability of $\chi'$ is immediate by the transportability of $\chi$.
\end{proof}

\subsection{Proof of equivalence}

We now show that instead of amplimorphisms, we can equivalently talk about endomorphisms.
For any cone $\Lambda$ and any $n \in \N$, fix a row vector $\Iso(\Lambda, n) := (V_1, \cdots, V_n)$ whose components are isometries $V_i \in \caR(\Lambda)$ satisfying $V_i^* V_j = \delta_{ij} \1$ and $\sum_{i = 1}^n V_i V_i^* = \1_n$.
(Since $\caR(\Lambda)$ is an infinite factor, we can repeatedly apply the halving lemma~\cite[Lemma 6.3.3]{KadisonR97} to obtain such isometries).
For any $\chi \in \Amp$ fix an allowed cone $\Lambda_{\chi}$ such that $\chi$ is localized on $\Lambda_{\chi}$ and write $\Iso_{\chi} = \Iso(\Lambda_{\chi}, n)$, where $n$ is the degree of $\chi$.

Now let $\chi \in \Amp$ be a unital amplimorphism of degree $n$. We define $\nu_{\chi} : \caB \rightarrow \caB$ to be the endomorphism given by
\begin{equation}
    \nu_{\chi}(O) := \Iso_{\chi} \, \chi(O) \Iso_{\chi}^*.
\end{equation}
Here we see $\Iso_{\chi}^*$ as a column vector with entries $V_i^*$.
One easily verifies that this indeed is an endomorphism and that $\nu_{\chi}$ is localized in $\Lambda_{\chi}$.

If $\chi, \chi' \in \Amp$ are unital amplimorphisms and $T \in (\chi | \chi')$, we define $t_T \in \caB(\caH_0)$ by $t_T = \Iso_{\chi} T \Iso_{\chi'}^*$. Then
\begin{equation}
    t_T \nu_{\chi'}(O) = \Iso_{\chi} \, T \, \Iso_{\chi'}^* \, \Iso_{\chi'} \, \chi'(O) \, \Iso_{\chi'}^* = \Iso_{\chi} \, T \, \chi'(O) \, \Iso_{\chi'}^* = \Iso_{\chi} \, \chi(O) \, T \, \Iso_{\chi'}^* = \nu_{\chi}(O) \, t_T
\end{equation}
so $t_T \in (\nu_{\chi} | \nu_{\chi'})$. The map $T \mapsto t_T$ defines a *-isomorphism of intertwiner spaces $(\chi | \chi')$ and $(\nu_{\chi} | \nu_{\chi'})$.

It follows in particular that the $\nu_{\chi}$ obtained in this way are transportable. Indeed, let $\Lambda'$ be some cone. By transportability of $\chi$ and Corollary \ref{cor:ampli is equivalent to unital ampli in the same cone} there is unital $\chi' \in \Amp$ localized on $\Lambda'$ and a unitary $U \in (\chi | \chi')$. Then $t_U \in (\nu_{\chi} | \nu_{\chi'})$ is also unitary.

Since $\Iso_{\chi} \in ( \nu_{\chi} | \chi )$ is an equivalence of amplimorphisms, we conclude in particular that every unital amplimorphsm in $\Amp$ is equivalent to an endomorphism in $\DHR$. Together with Corollary \ref{cor:ampli is equivalent to unital ampli in the same cone} we obtain the following lemma.

\begin{lemma} \label{lem:ampli is equivalent to endo}
Every $\chi \in \Amp$ is equivalent to an endomorphism $\rho_\chi$ in the subcategory $\DHR$. 
\end{lemma}

Even though we do not need it to prove Theorem \ref{thm:main result}, we can now easily obtain the following proposition which says that the localized and transportable amplimorphisms are equivalent to the endomorphisms studied in~\cite{Ogata2021}.
\begin{proposition}
\label{prop:braided monoidal equivalence of Amp and DHR}
    $\DHR$ and $\Amp$ are equivalent as braided $\rm C^*$-tensor categories.
\end{proposition}
\begin{proof}
    Let $F : \DHR \rightarrow \Amp$ be the embedding functor. Clearly $F$ is linear, fully faithful, braided monoidal, and respects the $*$-structure. It remains to check that $F$ is essentially surjective, but this is immediate from Lemma \ref{lem:ampli is equivalent to endo}.
\end{proof}


 \section{Amplimorphisms from ribbon operators} \label{sec:amplimorphism from ribbon operators}

In this section we construct for each half-infinite ribbon $\rho$ a full subcategory $\Amp_{\rho}$ of $\Amp$ whose objects are constructed as limits of certain `ribbon operators' taking unitary representations of $\caD(G)$ as input. (See Appendix \ref{app:ribbon operators} for the definition and basic properties of ribbons and ribbon operators).
From the equivalence of the localized and transportable amplimorphisms to DHR endomorphisms, this amounts to explicitly constructing examples of representations that satisfy the superselection criterion.
More importantly, we can also define the intertwiners as (weak operator) limits of elements in the quasi-local algebra.
In the notation of~\cite{Ogata2021}, this amounts to finding explicit examples of the maps $T$ defined there, as well as how they act on the intertwiners.

The very concrete description of $\Amp_{\rho}$ and its intertwiners will allow us to identify the braiding and fusion in this category.
We will use this to show in Section \ref{subsec:Amp_rho and Rep} that the categories $\Amp_{\rho}$ are equivalent to $\Rep_f\, \caD(G)$ as braided $\rm C^*$-tensor categories, and in Section \ref{subsec:Amp_rho and Amp_f} that they are equivalent to the whole of $\Amp_f$, thus establishing the equivalence of $\Amp_f$ and $\Rep_f\, \caD(G)$ as braided $\rm C^*$-tensor categories.

\subsection{Finite ribbon multiplets}

Throughout the rest of this manuscript the tensor product $\otimes$ of two matrices over $\caA$ will always mean the usual matrix tensor product, while the tensor product $\otimes$ of an element of $\caA$ with a matrix over $\C$ means the amplifying tensor product, yielding a matrix over $\caA$.

\begin{definition} \label{def:simple ribbon operators}
	For any $n$-dimensional unitary representation $D$ of $\caD(G)$ and any ribbon $\rho$ define $\bF^D_{\rho} \in M_n(\caA)$ by
	\begin{equation}
		\bF_{\rho}^D = \sum_{g, h} \, F_{\rho}^{g, h} \otimes D \big( g, h \big).
	\end{equation}
\end{definition}

\begin{proposition} \label{prop:properties of simple ribbon operators}
	Let $\rho$ be a ribbon such that $s_i = \partial_i \rho$, $i = 1, 2$ have distinct vertices and faces, and let $D$ be an $n$-dimensional unital unitary representation of $\caD(G)$.
	\begin{enumerate}[label=(\roman*)]
		\item We have
            \begin{equation}
                \bF_{\rho}^{D} \cdot ( \bF_{\rho}^{D} )^* = (\bF_{\rho}^{D} )^* \cdot \bF_{\rho}^{D} = \1_n.
            \end{equation}
            In other words, $\bF_{\rho}^{D}$ is a unitary element of $M_n(\caA)$. \label{ribbon prop:unitarity}
		
            \item We have $\bF_{\bar \rho}^{D} = ( \bF_{\rho}^D )^*$. \label{ribbon prop:adjoint}
        
            \item \label{ribbon prop:sum and product} Let $D_1, D_2$ be unitary representations of $\caD(G)$. The direct sum and product of 
            ribbon operators $\bF_{\rho}^{D_1}$ and $\bF_{\rho}^{D_2}$ satisfy
    	\begin{equation}
    		\bF_{\rho}^{D_1} \oplus \bF_{\rho}^{D_2} =  \bF_{\rho}^{D_1 \oplus D_2}, \quad \bF_{\rho}^{D_1} \otimes \bF_{\rho}^{D_2} = \bF_{\rho}^{D_1 \times D_2}
            \end{equation}
            where the direct sum and tensor product on the left hand sides are the usual direct sum and tensor product of matrices (with $\caA$-valued components), and $D_1 \times D_2$ is the monoidal product of the two representations (see Appendix~\ref{app:introduction to D(G)}).
            
		\item If $\rho = \rho_1\rho_2$ then
            \begin{equation}
                \bF_{\rho}^D = \bF_{\rho_1}^D \cdot \bF_{\rho_2}^D.
            \end{equation} \label{ribbon prop:concatenation}
            
		\item If $t \in (D_1 | D_2)$ then
            \begin{equation}
                 \bF_{\rho}^{D_1} (\1 \otimes t) = (\1 \otimes t) \bF_{\rho}^{D_2}, \quad (\bF_{\rho}^{D_1})^* (\1 \otimes t) = (\1 \otimes t) (\bF_{\rho}^{D_2})^* .
            \end{equation} \label{ribbon prop:intertwiners}
		\item If $\rho_1$ and $\rho_2$ are positive ribbons with common initial site $s_0$ as in Figure \ref{fig:braiding positive ribbons}, then
            \begin{equation} \label{eq:braiding of positive ribbon multiplets}
                \bF_{\rho_2}^{D_2} \otimes \bF_{\rho_1}^{D_1} = (\I \otimes B(D_1, D_2)) \cdot ( \bF_{\rho_1}^{D_1} \otimes \bF_{\rho_2}^{D_2} ) \cdot (\I \otimes P_{12}).
            \end{equation} \label{ribbon prop:braiding}
            where $B(-, -)$ is the braiding on $\Rep_f \caD(G)$, and $P_{12}$ interchanges the factors in the tensor product of the representation spaces of $D_1$ and $D_2$ (see Appendix \ref{app:introduction to D(G)}).
	\end{enumerate}
\end{proposition}

\begin{proof}
    By straightforward computations using Eqs. \eqref{eq:ribbon multiplication and adjoint}, \eqref{eq:ribbon reversal}, \eqref{eq:sum to identity}, and using the braid relation \eqref{eq:braiding positive ribbons} to obtain item \ref{ribbon prop:braiding}.
\end{proof}

\begin{figure}[!ht]
\centering
\includegraphics[width = 0.4\textwidth]{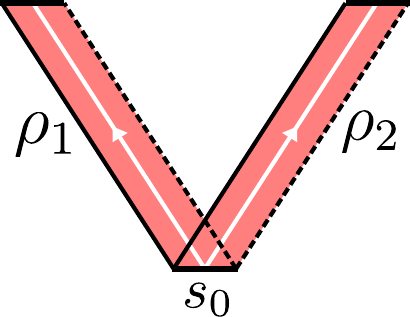}
\caption{Braiding positive ribbon operators, both having the same starting site $s_0$.}
\label{fig:braiding positive ribbons}
\end{figure}

\subsection{Amplimorphisms of the quasi-local algebra from ribbon multiplets}

\subsubsection{Construction}

For any finite ribbon $\rho$ and any $n$-dimensional unitary representation $D$ of $\caD(G)$, define linear maps $\mu_{\rho}^{D} : \caA \rightarrow M_n(\caA) \simeq \caA \otimes M_n(\C)$ by
\begin{equation}
	\mu_{\rho}^{D}(O) := \bF_{\rho}^{D} \cdot (O \otimes \I_n) \cdot (\bF_{\rho}^{D})^*.
\end{equation}
Note that by Proposition~\ref{prop:properties of simple ribbon operators} it follows directly that $\mu_\rho^D$ is a $*$-homomorphism.

A half-infinite ribbon $\rho = \{\tau_n\}_{n = 1}^{\infty}$ is a sequence of triangles labelled by $n \in \N$ such that $\partial_1 \tau_n = \partial_0 \tau_{n+1}$ for all $n\in \bbN$ and such that no edge of the lattice belongs to more than one of these triangles.

For any half-infinite ribbon $\rho = \{\tau_n\}$, denote by $\rho_n$ the ribbon consisting of the first $n$ triangles of $\rho$ and by $\rho_{>n} = \rho \setminus \rho_n$ the half-infinite ribbon obtained from $\rho$ by omitting the first $n$ triangles. Then a standard argument using Proposition~\ref{prop:properties of simple ribbon operators}\ref{ribbon prop:concatenation} shows the following limiting maps are well defined.
\begin{definition} \label{def:irreducible amplimorphism}
	For any half-infinite ribbon $\rho$ and any $n$-dimensional unitary representation $D$ of $\caD(G)$, define a linear map $\mu_{\rho}^{D} : \caA \rightarrow M_n(\caA)$ by
	\begin{equation}
		\mu_{\rho}^{D}(O) := \lim_{n \uparrow \infty} \, \mu_{\rho_n}^{D}(O).
	\end{equation}
\end{definition}

We have:
\begin{lemma}[Lemma 5.2 of \cite{Naaijkens2015}] \label{lem:ampli properties}
	The map $\mu_{\rho}^{D}:\caA \rightarrow M_n(\caA)$ is a unital *-homomorphism. \ie it is an amplimorphism of $\caA$ of degree $n$. Moreover, if the support of $O \in \cstar$ is disjoint from the support of $\rho$ then $\mu_{\rho}^{D}(O) = O \otimes \I_{n}$. For any $O \in \caA^{\loc}$ we have $\mu_{\rho}^{D}(O) = \mu_{\rho_n}^{D}(O)$ for all $n$ large enough. 
\end{lemma}

For each site in the model, it is possible to define an action $\gamma : \caD(G) \to \operatorname{Aut}(\alg{A})$ of the quantum double Hopf algebra.
The amplimorphisms constructed here transform covariantly with respect to this action.
These transformation properties (and of the ribbon multiplets themselves under this action) are essentially what connects these amplimorphisms to representations of $\caD(G)$.
For our purposes it is not necessary to spell out the details, and we refer the interested reader to~\cite{HamdanThesis}.

\subsubsection{Direct sum and tensor product}
The direct sum and tensor product of amplimorphisms of $\cstar$ are defined in the same way as amplimorphisms of $\caB$. We have for all $O \in \cstar$,
\begin{equation}
	(\mu_1 \times \mu_2)^{u_1 u_2, v_1 v_2}(O) = \mu_1^{u_1 v_1} \big(  \mu_2^{u_2 v_2}(O)   \big),
\end{equation}
and the direct sum of $\mu_1 : \cstar \rightarrow M_m(\cstar)$ and $\mu_2 : \cstar \rightarrow M_n(\cstar)$ is the amplimorphism $\mu_1 \oplus \mu_2 : \cstar \rightarrow M_{m + n}(\cstar)$ that maps $O \in \cstar$ to the block diagonal matrix with blocks $\mu_1(O)$ and $\mu_2(O)$.

\begin{lemma} \label{lem:sum and product of ribbon amplimorphisms}
	If $\rho$ is a finite or half-infinite ribbon then
	\begin{equation}
		\mu_{\rho}^{D_1} \oplus \mu_{\rho}^{D_2} = \mu_{\rho}^{D_1 \oplus D_2}, \quad \mu_{\rho}^{D_1} \times \mu_{\rho}^{D_2} = \mu_{\rho}^{D_1 \times D_2}.
	\end{equation}
\end{lemma}

\begin{proof}
	First consider the case where  $\rho$ is a finite ribbon.
    For ease of notation we omit the subscripts $\rho$ in the following.
    For any $O \in \caA$ we have
	\begin{align*}
		( \mu^{D_1} \oplus \mu^{D_2} )(O) &= \mu^{D_1}(O) \oplus \mu^{D_2}(O) = \bF^{D_1} (O \otimes \I_{n_1}) (\bF^{D_1})^* \oplus \bF^{D_2} (O \otimes \I_{n_2}) ( \bF^{D_2} )^* \\
						  &= \big( \bF^{D_1} \oplus \bF^{D_2} \big) \, (O \otimes \I_{n_1 + n_2} ) \, \big( \bF^{D_1} \oplus \bF^{D_2} \big)^* = \mu^{D_1 \oplus D_2}(O),
	\end{align*}
	where the last step uses item~\ref{ribbon prop:sum and product} of Proposition \ref{prop:properties of simple ribbon operators}.

	For the product, we compute componentwise
	\begin{align*}
		(\mu^{D_1} \times \mu^{D_2})(O)^{u_1 u_2; v_1 v_2} &= \mu^{D_1; u_1 v_1} \big(  \mu^{D_2; u_2 v_2}(O) \big) = \sum_{w_2} \, \mu^{D_1 ; u_1 v_1} \left( \bF^{D_2 ; u_2 w_2} \, O \, (\bF^{D_2 ; v_2 w_2})^* \right) \\
								   &= \sum_{w_1, w_2} \, \bF^{D_1; u_1 w_1} \, \bF^{D_2; u_2 w_2} \, O \, (\bF^{D_2; v_2 w_2})^* \, (\bF^{D_1;v_1 w_1})^* \\
								   &= \sum_{w_1, w_2} \, ( \bF^{D_1} \times \bF^{D_2} )^{u_1 u_2; w_1 w_2} \, O \, (( \bF^{D_1} \times \bF^{D_2} )^*)^{w_1 w_2 ; v_1 v_2} \\
								   &= \big( \bF^{D_1 \times D_2} \, (O \otimes \I_{n_1 n_2}) \, (\bF^{D_1 \times D_2})^* \big)^{u_1 u_2; v_1 v_2} = \mu^{D_1 \times D_2}(O)^{u_1 u_2; v_1 v_2}
	\end{align*}
	where the next to last step again uses item~\ref{ribbon prop:sum and product} of Proposition \ref{prop:properties of simple ribbon operators}.

    If $\rho$ is half-infinite, then the claim follows from the finite case by taking the limit of $\mu_{\rho_n}^D$.
\end{proof}

\subsubsection{Transportability}

We would like to extend the $\mu_{\rho}^{D}$ to amplimorphisms of the allowed algebra $\caB$. To this end, we must first establish their transportability.

We begin with a basic lemma which shows in particular that if $\rho$ and $\rho'$ coincide eventually, then $\mu_{\rho}^{D}$ and $\mu_{\rho'}^{D}$ are unitarily equivalent. Recall that if $\rho$ is a half-infinite ribbon, $\rho_n$ denotes the finite ribbon consisting of the first $n$ triangles of $\rho$, and $\rho_{> n}$ denotes the half-infinite ribbon obtained from $\rho$ by removing its first $n$ triangles. In particular, $\rho = \rho_n \rho_{>n}$.

\begin{lemma} \label{lem:ribbon shortening intertwiner}
	Let $\rho$ be a half-infinite positive ribbon and let $D$ be an $n$-dimensional unitary representation of $\caD(G)$. Then
	\begin{equation}
		\mu_{\rho}^{D} = \Ad[ \bF_{\rho_n}^{D} ] \circ \mu_{\rho_{>n}}^{D}
	\end{equation}
	for any $n \in \N$.
\end{lemma}

\begin{proof}
	This follows immediately from the definitions, Lemma \ref{lem:ampli properties}, and Proposition \ref{prop:properties of simple ribbon operators}.
\end{proof}

Since the $\bF_{\rho_n}^D$ are unitary operators, this establishes transportability over a finite distance.
To construct more general intertwiners, we need to use a limiting procedure.
\begin{definition} \label{def:bridge}
    Let $\rho$ and $\rho'$ be two half-infinite ribbons. A sequence of finite ribbons $\{ \xi_n \}_{n \in \N}$ is said to be a \emph{bridge} from $\rho$ to $\rho'$ if for each $n$ the concatenations $\sigma_n = \rho_n \xi_n \bar \rho'_n$ are finite ribbons and the bridges $\xi_n$ are eventually supported outside any ball. We call $\{ \sigma_n \}$ the intertwining sequence of the bridge $\{ \xi_n \}.$

    We say a half-infinite ribbon $\rho$ is `good' if it is supported in a cone $\Lambda$ and for any other cone $\Lambda'$ that is disjoint from $\Lambda$, there is a half-infinite ribbon $\rho'$ and a bridge from $\rho$ to $\rho'$. Note that any cone contains plenty of good half-infinite ribbons, both positive and negative ones.
\end{definition}

\begin{lemma} \label{lem:construction of transporters}
	Let $\rho$ be a half-infinite positive ribbon and let $\rho'$ be half-infinite negative ribbon both supported in a cone $\Lambda$ and with initial sites $s, s'$ respectively. Suppose there is a bridge from $\rho$ to $\rho'$ with intertwining sequence $\{ \sigma_m = \rho_m \xi_m \overline \rho'_m \}$ all supported in $\Lambda$. Let $D$ be an $n$-dimensional unitary representation of $\caD(G)$. Then there is a unitary $U \in M_n(\caR(\Lambda))$ such that
	\begin{equation} \label{eq:transporter property}
		(\pi_0 \otimes \id_n) \circ \mu_{\rho'}^{D} = \Ad[ U ] \circ (\pi_0 \otimes \id_n) \circ \mu_{\rho}^{D}.
	\end{equation}
\end{lemma}

\begin{proof}
	Consider the family of half-infinite ribbons $\rho^{(m)} = \rho'_m \overline{\xi_m} \rho_{>m}$, see Figure \ref{fig:bridge}. We first show that
	\begin{equation} \label{eq:finite intertwining}
		\mu_{\rho^{(m)}}^{D} = \Ad[ \bF_{\bar \sigma_m}^{D}  ] \circ \mu_{\rho}^{D}.
	\end{equation}
	Indeed, by Proposition \ref{prop:properties of simple ribbon operators} we have $\bF_{\bar \sigma_m}^{D} = \big( \bF_{\rho_m}^{D} \cdot \bF_{\xi_m}^{D} \cdot \bF_{\overline \rho'_m}^{D}  \big)^* = \bF_{\rho'_m}^{D} \cdot \bF_{\overline{\xi_m}}^{D} \cdot ( \bF_{\rho_m}^{D} )^*$ so for any $O \in \caA^{\loc}$ we have
	\begin{align*}
		\big( \Ad[ \bF_{\bar \sigma_m}^{D} ] \circ \mu_{\rho}^{D} \big)(O) &= \lim_{N \uparrow \infty} \, \Ad \left[ \bF_{\rho'_m}^{D} \cdot \bF_{\overline{\xi_m}}^{D} \cdot (\bF_{\rho_m}^{D})^* \cdot \bF_{\rho_m}^{D} \cdot \bF_{(\rho_{>m})_N}^{D} \right] (O \otimes \I_n). \\
    \intertext{Now we use unitarity to get}
         &= \lim_{N \uparrow \infty} \, \Ad \left[ \bF_{\rho'_m}^{D} \cdot \bF_{\overline{\xi_m}}^{D} \cdot \bF_{(\rho_{>m})_N}^{D} \right] (O \otimes \I_n) = \mu_{\rho^{(m)}}(O)
	\end{align*}
	as required. 
	
	\begin{figure}[!t]
	\centering
	\includegraphics[width = 0.4\textwidth]{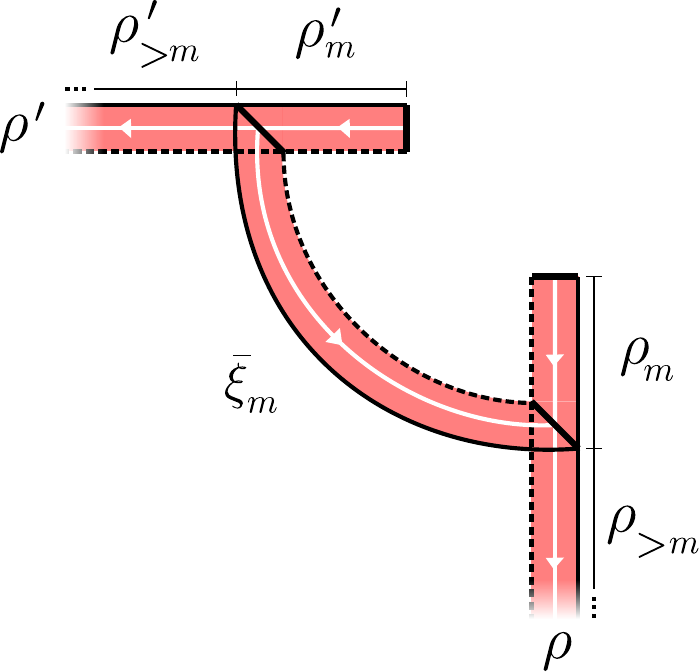}
	\caption{The finite ribbon $\overline{\xi_m}$ is a bridge from ribbon $\rho_m'$ to $\rho_m$.} 
    \label{fig:bridge}
	\end{figure}

	By Lemma \ref{lem:unitary transporters} the components of the image of  $\bF_{\bar \sigma_n}^{D}$ under $\pi_0 \otimes \id_n$ converge in the strong-* topology, and therefore so does the full image of $\bF_{\bar \sigma_n}^{D}$. Denote the limit by $U$. Since the $\bF_{\bar \sigma_n}^{D}$ are all unitary (Proposition \ref{prop:properties of simple ribbon operators}) it follows from Lemma \ref{lem:abstract unitarity} that $U$ is unitary. Since all the $\bF_{\bar \sigma_n}^{D}$ are supported in the cone $\Lambda$, it follows that $U \in M_n(\caR(\Lambda))$.
    
    Let $O \in \caA^{\loc}$. Then
    \begin{align*}
        U \cdot (\pi_0 \otimes \id_n) \big( \mu^D_{\rho}(O) \big) &= \lim_{n \uparrow \infty} \, (\pi_0 \otimes \id_n) \big( \bF_{\bar \sigma_n}^{D} \cdot \mu^D_{\rho}(O) \big) \\
        &= \lim_{n \uparrow \infty} \, (\pi_0 \otimes \id_n) \big(  \mu^D_{\rho^{(n)}}(O) \cdot \bF_{\bar \sigma_n}^{D}  \big) \\
        &= (\pi_0 \otimes \id_n) \big( \mu^D_{\rho'}(O) \big) \cdot U
    \end{align*}
    where we used componentwise continuity of multiplication in the strong operator topology in the first equality, Eq.~\eqref{eq:finite intertwining} to obtain the second equality, and the fact that $\mu^D_{\rho^{(n)}}(O) = \mu^D_{\rho'}(O)$ for $n$ large enough and again componentwise continuity of multiplication to obtain the last equality. Since $\caA^{\loc}$ is dense in $\caA$, we conclude that Eq. \eqref{eq:transporter property} holds, which completes the proof.
\end{proof}

\begin{remark}
This answers a question that was left open in~\cite{Naaijkens2015}, namely the construction of unitary charge transporters that transport charges between two cones, and not just over a finite distance.
Note that Lemma~\ref{lem:construction of transporters} implies that the representation $(\pi_0 \otimes \id_n) \circ \mu_{\rho}^D$ satisfies a variant of the superselection criterion, where we have $(\pi_0 \otimes \id_n) \circ \mu_{\rho}^D \upharpoonright \caA_{\Lambda^c} \cong n \cdot \pi_0 \upharpoonright \caA_{\Lambda^c}$.
That is, instead of unitary equivalence as in~\eqref{eq:sselect}, we have \emph{quasi-}equivalence.
As we shall see shortly, in the case at hand the two notions can be seen to coincide.
\end{remark}

\subsection{Amplimorphisms of the allowed algebra from ribbon multiplets}

The transportability of the $\mu_{\rho}^D$ established above in Lemma \ref{lem:construction of transporters} allows us to extend these amplimorphsisms to localized and transportable amplimorphisms of the allowed algebra $\caB$.

\begin{proposition} \label{prop:Amp is nonempty}
	Let $\rho$ be a good half-infinite positive ribbon that is contained in an allowed cone $\Lambda$, then there exists a unique amplimorphism $\chi_{\rho}^{D} : \caB \rightarrow M_n(\caB)$ whose restriction to $\caR(\Lambda)$ is weakly continuous, and satisfies 
	\begin{equation}
		\chi_{\rho}^{D} \circ \pi_0(O) = (\pi_0 \otimes \id_n) \circ \mu_{\rho}^{D}(O).
	\end{equation}
    for all $O \in \cstar$. Moreover, $\chi_{\rho}^{D}$ is localized in $\Lambda$ and is transportable. It is therefore an object of $\Amp$.
\end{proposition}

\begin{proof}
    Recall that $\caB$ is a direct limit of cone algebras $\caR(\Lambda)$.
    Note that $\mu_\rho^D$ restricts to an amplimorphism $\cstar[\Lambda] \to M_n(\cstar[\Lambda])$.
    We show that we can extend this (on both sides) to $\caR(\Lambda)$.
    This construction is compatible with the direct structure on the set of allowed cones, and hence defines an amplimorphism of $\caB$.

    To see that we can extend $\mu_\rho^D$ (restricted to $\cstar[\Lambda]$) to $\chi_\rho^D : \caR( \Lambda) \rightarrow M_n(\caR(\Lambda))$,
    note first that for every $\Lambda$ we have the existence of a forbidden cone $\widehat \Lambda$ disjoint from $\Lambda$.
    Since $\rho$ is good and by Lemma \ref{lem:construction of transporters}, we have that $\mu_\rho^D \simeq \mu_{\hat \rho}^D$ where $\hat \rho$ is localized in $\widehat \Lambda$.
    Let $U$ be the unitary implementing this equivalence.
    By locality we have that for all $O \in \cstar[ \Lambda]$, it holds that $\mu_{\widehat \rho}^D(O) = O \otimes \I_n$.
    
    Define $\chi_\rho^D(O) := \Ad[U](O \otimes \I_n)$ for all $O \in \caR(\Lambda)$.
    By construction, it follows that $\chi_\rho^D(O) = \mu_\rho^D(O)$ for all $O \in \cstar[\Lambda]$.
    Let $O \in \caR(\Lambda)$. 
    Then there exist $\cstar[\Lambda] \ni O_\lambda \to O$ weakly since $\cstar[\Lambda]$ is weak-operator dense in $\caR(\Lambda)$.
    Hence we have
    \[
        \lim_\lambda \mu_\rho^D(O_\lambda) = \lim_\lambda \Ad[U](O_\lambda \otimes \I_n) = \Ad[U](O \otimes \I_n) = \chi_\rho^D(O),
    \]
    where all limits are in the weak operator topology and we used that  $\Ad[U]$ is weakly continuous.
    Hence, $\chi_\rho^D$ is uniquely determined by $\mu_\rho^D$.
    This action on $\caR(\Lambda)$ is independent of the choice of forbidden cone $\widehat \Lambda$, so the extensions to $\caR(\widetilde \Lambda)$ for different cones are consistent with each other. These actions therefore define a *-homomorphism $\chi_\rho^D$ on all of $\caB$.

    Now consider some $O \in \cstar[\Lambda^c]^{\loc}$. Then there is a forbidden cone $\widehat \Lambda$, disjoint from $\Lambda$ and such that $O \in \caA_{\widehat \Lambda^c}$. Let $\mu_{\hat \rho}^D$ and $U$ be as above.
    We have 
    \begin{align}
        \chi_\rho^D(O) = U(O \otimes \I_n) U^* = U \mu_{\hat \rho}^D (O) U^* = \mu_\rho^D(O) = O \otimes \I_n.
    \end{align}
    Since this holds for any $O \in \caA_{\Lambda^c}^{\loc}$, we find that $\chi_\rho^D$ is localized in $\Lambda$.

    Now consider an allowed cone $\widetilde \Lambda$. Using transportability of $\mu_\rho^D$ (Lemma~\ref{lem:construction of transporters}) we have that there exists some $\mu_{\widetilde{\rho}}^D \simeq \mu_\rho^D$ localized in $\widetilde \Lambda$. Uniquely extend $\mu_{\widetilde\rho}^D$ to $\chi_{\widetilde \rho}^D$ as above. Then any unitary intertwiner from $\mu_{\widetilde \rho}^D$ to $\mu_{\rho}^D$ is an equivalence between $\chi_\rho^D$ and $\chi^D_{\widetilde\rho}$, showing that $\chi_\rho^D$ is indeed transportable. 
\end{proof}

This proposition allows the following definition.
\begin{definition} \label{def:amplis from ribbon operators}
    Let $\rho$ be a good half-infinite ribbon and $D$ a unitary representation of $\caD(G)$. Then we denote by $\chi_{\rho}^D$ the unique amplimorphism of $\caB$ that satsifies the properties of Proposition~\ref{prop:Amp is nonempty}.
\end{definition}

\begin{lemma} \label{lem:sum and product of amplimorphism representations}
	For any good half-infinite ribbon $\rho$ supported in an allowed cone we have
	\begin{equation}
		\chi_{\rho}^{D_1} \oplus \chi_{\rho}^{D_2} = \chi_{\rho}^{D_1 \oplus D_2}, \quad \chi_{\rho}^{D_1} \times \chi_{\rho}^{D_2} = \chi_{\rho}^{D_1 \times D_2}. 
	\end{equation}
\end{lemma}

\begin{proof}
	Follows immediately from Lemma \ref{lem:sum and product of ribbon amplimorphisms} and the uniqueness of the $\chi_{\rho}^D$ as extensions of the $\mu_{\rho}^D$.
\end{proof}

\subsection{Braided monoidal subcategory of \texorpdfstring{$\Amp$}{} on a fixed ribbon} \label{subsec:subcategory of amplimorphisms on a fixed ribbon}

We will call a half-infinite ribbon $\rho$ allowed if it is supported in some allowed cone. Let $\rho$ be a positive good allowed half-infinite ribbon and let $\Amp_{\rho}$ be the full subcategory of $\Amp$ whose objects are the localized and transportable amplimorphsisms $\chi_{\rho}^D$ for arbitrary unitary representations $D$. Lemma \ref{lem:sum and product of amplimorphism representations} shows that this subcategory is closed under direct sums and tensor products, so $\Amp_{\rho}$ is a full monoidal subcategory of $\Amp$. Being closed under the tensor product, the subcategory $\Amp_{\rho}$ inherits the braiding of $\Amp$ defined in Section \ref{subsec:braiding}. Finally, it follows from Proposition \ref{prop:T to t} below that $\Amp_{\rho}$ has subobjects, so it is in fact a full braided $\rm C^*$-tensor subcategory of $\Amp$.



\section{Simple objects of \texorpdfstring{$\Amp$}{Amp}} \label{sec:simples of Amp}

In the previous section we constructed full subcategories $\Amp_{\rho}$ of $\Amp$ whose objects are constructed from unitary representations of $\caD(G)$. These subcategories will play a crucial role in establishing the equivalence of $\Amp_f$ and $\Rep_f \caD(G)$. 

In order to do this we must first establish that the amplimorphisms $\chi_{\rho}^D$ are finite, so that they belong to $\Amp_f$. Then we must show that $\chi_{\rho}^D$ is a simple object whenever $D$ is an irreducible representation. Conversely, we must show that any simple object of $\Amp$ is equivalent to an amplimorphism $\chi^D_{\rho}$ for some irreducible representation $D$. In this section we prove these facts by appealing to the classification of irreducible anyon sectors of Kitaev's quantum double models achieved in \cite{bols2023classificationanyonsectorskitaevs}, which we first review.

\subsection{Classification of irreducible anyon sectors}

\begin{definition} \label{def:anyon representation}
    A *-representation $\pi : \caA \rightarrow \caB(\caH)$ is said to satisfy the superselection criterion with respect to the representation $\pi_0$ if for any cone $\Lambda$ there is a unitary $U : \caH_0 \rightarrow \caH$ such that
    \begin{equation*}
        \pi(O) = U \pi_0(O) U^*
    \end{equation*}
    for all $O \in \caA_{\Lambda^c}$. If $\pi$ is moreover irreducible, then we call $\pi$ an \emph{anyon representation}.
\end{definition}

The following theorem follows directly from Theorem 2.4 and Proposition 5.19 of \cite{bols2023classificationanyonsectorskitaevs}.
\begin{theorem}[\cite{bols2023classificationanyonsectorskitaevs}] \label{thm:anyon classification}
    Let $\rho$ be a good half-infinite ribbon. The representations $\chi_{\rho}^D \circ \pi_0$ are anyon representations if and only if $D$ is irreducible. Two such anyon representations $\chi_{\rho}^{D_1} \circ \pi_0$ and $\chi_{\rho}^{D_2} \circ \pi_0$ are unitarily equivalent (disjoint) whenever the irreducible representations $D_1$ and $D_2$ are equivalent (disjoint).

    Moreover, any anyon representation $\pi$ is unitarily equivalent to $\chi_{\rho}^D \circ \pi_0$ for some irreducible representation $D$.
 \end{theorem}

\subsection{Simple amplimorphisms}

Fix a good allowed half-infinite ribbon $\rho$.

\begin{proposition} \label{prop:simple objects from irreducible representations}
    Let $D_1$ and $D_2$ be irreducible representation of $\caD(G)$. Then the amplimorphisms $\chi_{\rho}^{D_1}$ and $\chi_{\rho}^{D_2}$ are simple objects of $\Amp$. If they are equivalent, then the representations $D_1$ and $D_2$ must be equivalent.
\end{proposition}

The converse to the second part, namely that $\chi_{\rho}^{D_1}$ and $\chi_{\rho}^{D_2}$ are equivalent if $D_1$ and $D_2$ are equivalent will be shown later in Proposition \ref{prop:t to T}. \\
 
\begin{proof}
    Suppose $\chi_{\rho}^{D_1}$ were not simple. Then there is a non-trivial orthogonal projector $p \in ( \chi_{\rho}^{D_1} | \chi_{\rho}^{D_1})$. Since $\chi_{\rho}^{D_1}$ is unital, this implies
    \begin{equation*}
        p \, \cdot \,  (\chi_{\rho}^{D_1} \circ \pi_0)(O) =  (\chi_{\rho}^{D_1} \circ \pi_0)(O) \, \cdot \,  p \quad \text{for all} \, O \in \caA.
    \end{equation*}
    But this shows that $p$ is in the commutant of the representation $\chi_{\rho}^{D_1} \circ \pi_0$. Since the latter representation is irreducible by Theorem \ref{thm:anyon classification}, $p$ cannot be a non-trivial projection. We conclude that $\chi_{\rho}^{D_1}$ is simple.

    Similarly, if $U \in (\chi_{\rho}^{D_2} | \chi_{\rho}^{D_1})$ is a unitary equivalence of unital amplimorphisms then $U$ is also a unitary intertwiner of representations $\chi_{\rho}^{D_1} \circ \pi_0$  and  $\chi_{\rho}^{D_2} \circ \pi_0$. By Theorem \ref{thm:anyon classification} such a $U$ can exists only if $D_1$ and $D_2$ are equivalent.
\end{proof}

\begin{proposition} \label{prop:charcterization of simple amplis}
    Any simple object of $\Amp$ is equivalent to $\chi_{\rho}^D$ for some irreducible representation $D$.
\end{proposition}

\begin{proof}
    Let $\chi$ be a simple amplimorphism of degree $n$. By Lemma \ref{lem:ampli is equivalent to endo} we can assume without loss of generality that $\chi$ is an endomorphism.
    
    Let us show that the *-representation $\chi \circ \pi_0 : \caA \rightarrow \caB(\caH_0)$ satisfies the superselection criterion, Definition \ref{def:anyon representation}. Let $\Lambda$ be a cone. By transportability there is an endomorphism $\chi' \in \DHR$ localized in $\Lambda^c$ such that $\chi \sim \chi'$. Let $U \in (\chi' | \chi)$ be a (necessarily unitary) equivalence. Then one has $(\chi \circ \pi_0)(O) = U^* \pi_0(O) U$ for any $O \in \caA_{\Lambda}$. Since $\Lambda$ was arbitrary, this shows that $\chi \circ \pi_0$ indeed statisfies the superselection criterion.

    We now use the assumption that $\chi$ is simple to show that $\chi \circ \pi_0$ is in fact an anyon representation. That is, we want to show that $\chi \circ \pi_0$ is irreducible. To obtain a contradiction, suppose $p \in \caB(\caH)$ is a non-trivial projection intertwining the representation $\chi \circ \pi_0$ with itself. Since commutation is preserved under weak limits, it follows that $p \in (\chi | \chi)$, contradicting simplicity of $\chi$. So $\chi \circ \pi_0$ is indeed an anyon representation.
    
    By Theorem \ref{thm:anyon classification} it follows that $\chi \circ \pi_0$ is unitarily equivalent as a $*$-representation of $\caA$ to $\chi_{\rho}^{D} \circ \pi_0$ for some irreducible representation $D$. Let $U$ be an intertwining unitary. It follows by continuity that in fact $U \in (\chi | \chi_{\rho}^D)$ is an equivalence of amplimorphisms, as required.
\end{proof}



\section{Equivalence of \texorpdfstring{$\Rep_f \caD(G)$}{RepD(G)}, \texorpdfstring{$\Amp_{\rho}$}{Amprho}, and \texorpdfstring{$\Amp_f$}{Ampf}} \label{sec:proof of main THM}

In this section we prove the remaining equivalences of categories needed to establish our main result, Theorem~\ref{thm:main result}.

\subsection{Equivalence of \texorpdfstring{$\Amp_{\rho}$}{} and \texorpdfstring{$\Rep_f \caD(G)$}{}} \label{subsec:Amp_rho and Rep}

Fix a good allowed half-infinite ribbon $\rho$. In this section we show that the category $\Amp_{\rho}$ introduced in Section \ref{subsec:subcategory of amplimorphisms on a fixed ribbon} is equivalent to $\Rep_f \caD(G)$, the category of finite dimensional unitary representations of $\caD(G)$.

\subsubsection{Monoidal equivalence} \label{subsec:monoidal equivalence}

Let us first show that for every intertwiner $t \in (D_1 | D_2 )$ of representations we can construct an intertwiner $T \in (\chi_{\rho}^{D_1} | \chi_{\rho}^{D_2})$ of amplimorphisms.

\begin{proposition} \label{prop:t to T}
	If $t \in (D_1 | D_2)$ then $T := \1 \otimes t \in (\chi_{\rho}^{D_1} | \chi_{\rho}^{D_2})$.
\end{proposition}
\begin{proof}
	For any $O \in \caA^{\loc}$ we have for all $n$ large enough (dropping $\pi_0$ from the notation)
    \[
    \begin{split}
		T \, \chi_{\rho}^{D_2}(O) &= T \, \mu_{\rho_n}^{D_2}(O) = (\1 \otimes t) \, \bF_{\rho_n}^{D_2} \, (O \otimes \I_n) \, (\bF_{\rho_n}^{D_2})^* \\
					  &= \bF_{\rho_n}^{D_1} \, (O \otimes \I_n) \, ( \bF_{\rho_n}^{D_1} )^* \, (\1 \otimes t) = \chi_{\rho}^{D_1}(O) \, T
                      \end{split}
    \]
	where we used item \ref{ribbon prop:intertwiners} of Proposition \ref{prop:properties of simple ribbon operators}. Let $\Lambda$ be an allowed cone containing $\rho$. Since $\cstar[\Lambda]^\loc$ is norm dense in $\cstar[\Lambda]$ which is in turn weakly dense in $\caR(\Lambda)$, using weak continuity of $\chi_\rho^{D_i}$ on cone algebras, this relation is true for all $O \in \caR(\Lambda)$. Since $\Lambda$ was an arbitrary allowed cone containig $\rho$, this relation holds for all $O \in \caB$. Thus $T \in (\chi_\rho^{D_1} | \chi_\rho^{D_2})$.
\end{proof}

Conversely, we want to show that all $T \in (\chi_{\rho}^{D_1} | \chi_{\rho}^{D_2})$ are of this form.
\begin{proposition} \label{prop:T to t}
	If $T \in (\chi_{\rho}^{D_1} | \chi_{\rho}^{D_2})$ then $T = \1 \otimes t$ for some $t \in (D_1 | D_2)$. In particular, the amplimorphisms $\chi_{\rho}^{D}$ are finite so $\Amp_{\rho}$ is a full $\rm C^*$-subcategory of $\Amp_f$.
\end{proposition}

\begin{proof}	
    Decompose $D_1$ and $D_2$ into direct sums of irreducibles (\cf Appendix \ref{app:introduction to D(G)}):
	\begin{equation}
		D_i \simeq \tilde D_i := \bigoplus_{r \in I} \, N_r^i \, \cdot \, D^{(r)},
	\end{equation}
	where $I$ is the finite set of equivalence classes of irreducible representations of $\caD(G)$ and $D^{(r)}$ is a representation in class $r$. Let $u_i \in ( D_i | \tilde D_i )$ be the unitaries implementing these equivalences. It follows from Proposition \ref{prop:t to T} that $U_i = (\1 \otimes u_i) \in (\chi_{\rho}^{D_i} | \chi_{\rho}^{\tilde D_i})$ and therefore $\tilde T := U_1^* T U_2 \in (\chi_{\rho}^{\tilde D_1} | \chi_{\rho}^{\tilde D_2})$.

    By Proposition \ref{prop:simple objects from irreducible representations},  $\{\chi_\rho^{D^{(r)}}\}_{r \in I}$ are disjoint simple objects of $\Amp_{\rho}$. Since the $\tilde D_i$ are direct sums of these it follows from Lemma \ref{lem:sum and product of amplimorphism representations} that the matrix blocks of $\tilde T$ mapping a $\chi_{\rho}^{D_r}$ subspace to a $\chi_{\rho}^{D_{r'}}$ are actually intertwiners of these amplimorphisms. It follows that the matrix blocks of $\tilde T$ corresponding to maps between copies of the same $\chi_{\rho}^{D_r}$ are multiples of the identity, and the other matrix blocks vanish, \ie $\tilde T = \1 \otimes \tilde t$ where
	\begin{equation}
		\tilde t = \bigoplus_{r} \tilde t_r \otimes \I_{n_r}
	\end{equation}
	with $\tilde t_r \in \Mat_{N^1_r \times N^2_r}(\C)$. Any such matrix $\tilde t$ belongs to $(\tilde D_1 | \tilde D_2)$. Since $u_i \in (D_i | \tilde D_i)$ it follows that $t = u_1 \tilde t u_2^* \in (D_1 | D_2)$. Now,
	\begin{equation}
		T = U_1 \tilde T U_2^* = (\1 \otimes u_1) (\1 \otimes \tilde t) (\1 \otimes u_2^*) = \1 \otimes t,
	\end{equation}
	which proves the claim.
\end{proof}

The two preceding propositions show that there is an isomorphsim between $(D_1 | D_2)$ and $( \chi_{\rho}^{D_1} | \chi_{\rho}^{D_2} )$ for all unitary representations $D_1, D_2$. We can use this isomorphisms to construct a monoidal equivalence between $\Rep_f \caD(G)$ and $\Amp_{\rho}$.

Consider the functor $F : \Rep_f \caD(G) \rightarrow \Amp_{\rho}$ which maps any unitary representation $D$ to the amplimorphism $\chi_{\rho}^{D}$, and maps any $t \in (D_1 | D_2)$ to $\1 \otimes t$. It follows from Proposition \ref{prop:t to T} that $F$ is indeed a functor. In fact, $F$ is linear and respects the $*$-structure. Moreover: 

\begin{proposition} \label{prop:monoidal equivalence}
	The functor $F : \Rep_f \caD(G) \rightarrow \Amp_{\rho}$ is a monoidal equivalence. In particular, $\Rep_f \caD(G)$ and $\Amp_{\rho}$ are equivalent as $\rm C^*$-tensor categories.
\end{proposition}

\begin{proof}
	Using Lemma \ref{lem:sum and product of amplimorphism representations} we find
	\begin{equation}
		F(D_1) \times F(D_2) = \chi_{\rho}^{D_1} \times \chi_{\rho}^{D_2} = \chi_{\rho}^{D_1 \times D_2} = F( D_1 \times D_2 ).
	\end{equation}
	Let $\id_{D_1, D_2} : F(D_1) \otimes F(D_2) \rightarrow F( D_1 \times D_2 )$ be the identity maps. Strict monoidality of $F$ means that the $\id_{D_1, D_2}$ form a natural transformation between functors $\times \circ (F, F) : \Rep_f \caD(G) \times \Rep_f \caD(G) \rightarrow \Amp_{\rho}$ and $F \circ \times :   \Rep_f \caD(G) \times \Rep_f \caD(G) \rightarrow \Amp_{\rho}$. Since  $\Amp_{\rho}$ is strict, this boils down to $F(t) \times F(t') = F(t \times t')$ for any $t \in (D_1 | D_2)$ and any $t' \in (D_1' | D_2')$, but this follows immediately from the definitions (recall in particular the definition in equation~\eqref{eq:tensor of intertwiners} of the tensor product of intertwiners of amplimorphisms).

	To see that $F$ is an equivalence of categories we note that $F$ is in fact an isomorphism, \ie $F$ is invertible with inverse $F^{-1}$ given on objects by $F^{-1}( \chi_{\rho}^D ) = D$ and on morphisms $T \in (\chi_{\rho}^{D_1} | \chi_{\rho}^{D_2})$ by $F^{-1}(T) = t$ with $t$ the unique intertwiner $t \in (D_1 | D_2)$ such that $T = \1 \otimes t$, \cf Proposition \ref{prop:T to t}.
\end{proof}

\subsubsection{Braided monoidal equivalence}

As remarked in Section \ref{subsec:subcategory of amplimorphisms on a fixed ribbon}, the subcategory $\Amp_{\rho}$ inherits the braiding of $\Amp$ defined in Section \ref{subsec:braiding}. Let us now compute the braiding between objects of $\Amp_{\rho}$ explicitly. 

In order to compute $\ep( \chi_{\rho}^{D_1} , \chi_{\rho}^{D_2} )$ we fix good negative half-infinite ribbons $\rho_{L}$ and $\rho_R$ as in Figure \ref{fig:explicit braiding}. By the proof of Lemma \ref{lem:construction of transporters} there are unitaries $U \in (\chi_{\rho_R}^{D_1} | \chi_{\rho}^{D_1})$ and $V \in (\chi_{\rho_{L}}^{D_2} | \chi_{\rho}^{D_2})$ that are limits in the strong-* operator topology of unitary sequences $U_n = \bF_{\overline \sigma_{R, n}}^{D_1}$ and $V_n = \bF_{\overline \sigma_{L, n}}^{D_2}$ with ribbons $\sigma_{L, n} = \overline \rho_{L, n} \xi_{L, n} \rho_n$ and $\sigma_{R, n} = \overline \rho_{R, n} \xi_{R, n} \rho_n$ as in Figure \ref{fig:explicit braiding}, so the ribbons $\{ 
\xi_{L/R, n} \}$ are bridges from $\rho$ to $\rho_{L, R}$.

\begin{figure}[!ht]
\centering
\includegraphics[width = 0.7\textwidth]{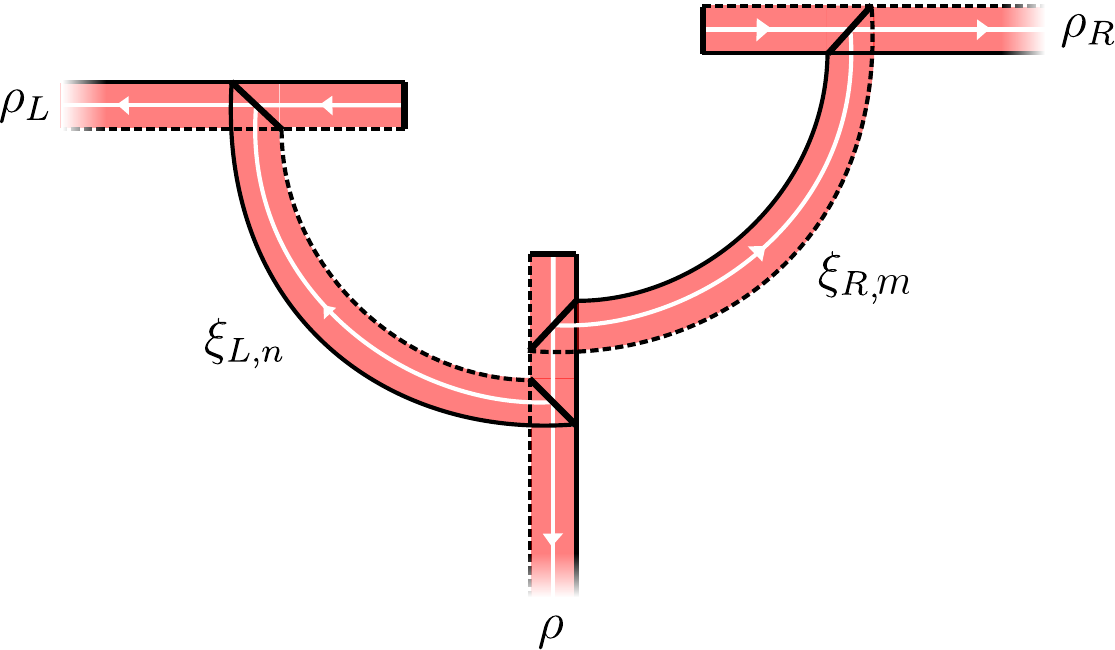}
\caption{The finite ribbon $\xi_{L,n}$ is a bridge from $\rho$ to $\rho_L$, and $\xi_{R,m}$ is a bridge from $\rho$ to $\rho_R$.}
\label{fig:explicit braiding}
\end{figure}

Let $\rho'_{L/R, n} = (\rho_{L/R})_{> n} \, \xi_{L/R, n} \, \rho_n$ and regard the unitaries $U_n$ and $V_n$ as intertwiners in $(\chi_{\rho_R}^{D_1} | \chi_{\rho'_{R, n}}^{D_1})$ and $(\chi_{\rho_L}^{D_2} | \chi_{\rho'_{L, n}}^{D_2})$ respectively, fix $l > 0$, and write $m = n + l$. Let $\zeta_{L/R, n} = \xi_{L/R, n} \overline \rho_{L/R, n}$ be such that $\sigma_{L, n} = \zeta_{L, n} \rho_n$ and $\sigma_{R, n} = \zeta_{R, n} \rho_n$.
Recall the braiding defined in equation~\eqref{eq:braiding of amplimorphisms defined}.
Noting that since all amplimorphisms are unital, the operator $P_{12}$ below used to define the braiding does not depend on $n$ or $m$, we have
\begin{align*}
	( V_{n+l}^* \times U_{n}^* )& \cdot P_{12} \cdot ( U_{n} \times V_{n+l} ) \\
         &= (V_m^* \otimes \1) \, \chi_{\rho_L}^{D_2}( U_n^* ) \cdot P_{12} \cdot  \chi_{\rho_R}^{D_1}( V_m ) \, (U_n \otimes \1) \\
             &= (V_m^* \otimes \1) (\1 \otimes U_n^*) \cdot P_{12} \cdot (\1 \otimes V_m) (U_n \otimes \1) \\
             &= ( \bF_{\rho_m}^{D_2} \otimes \1) ( \bF_{\zeta_{L, m}}^{D_2} \otimes \1 )  ( \1 \otimes \bF_{\rho_n}^{D_1}) ( \1 \otimes \bF_{\zeta_{R, n}}^{D_1} )  \cdot P_{12} \\
             & \quad\quad \cdot   ( \1 \otimes \bF_{\zeta_{L, m}}^{D_2} )^* ( \1 \otimes \bF_{\rho_m}^{D_2} )^* ( \bF_{\zeta_{R, n}}^{D_1} \otimes \1 )^* ( \bF_{\rho_n}^{D_1} \otimes \1 )^*\\
             \intertext{since $\rho_n$ is disjoint from the ribbons $\zeta_{L, n}$ and $\zeta_{R, n}$, and using item~\ref{ribbon prop:adjoint}  of Proposition \ref{prop:properties of simple ribbon operators} this becomes}
             &= ( \bF^{D_2}_{\rho_m} \otimes \bF^{D_1}_{\rho_n} ) ( \bF_{\bar \rho_m}^{D_2} \otimes \bF_{\bar \rho_n}^{D_1} ) \\
             & \quad \quad \cdot ( \bF_{\zeta_{L, m}}^{D_2} \otimes \bF_{\zeta_{R, n}}^{D_1} ) \cdot P_{12} \cdot ( \bF_{\zeta_{R, n}}^{D_1} \otimes \bF_{\zeta_{L, m}}^{D_2} )^* \\
             \intertext{using items \ref{ribbon prop:adjoint} and \ref{ribbon prop:sum and product} of Proposition \ref{prop:properties of simple ribbon operators} and unitarity, we get rid of the ribbon multiplets on $\rho_n, \rho_m$. The ribbons $\zeta_{R, n}$ and $\zeta_{L, m}$ are configured like the ribbons $\rho_1$ and $\rho_2$ of Figure \ref{fig:braiding positive ribbons} so we can apply item \ref{ribbon prop:braiding} of Proposition \ref{prop:properties of simple ribbon operators} to obtain}
			&= B(D_1, D_2).
\end{align*}

Since multiplication of operators is jointly continuous in the strong operator topology on bounded sets we have that
\begin{equation*}
    \ep( \chi_{\rho}^{D_1}, \chi_{\rho}^{D_2} ) = (V^* \times U^*) \cdot P_{12} \cdot (U \times V) = \lim_{n \uparrow \infty} ( V_{n+l}^* \times U_{n}^* ) \cdot P_{12} \cdot ( U_{n} \times V_{n+l} ).
\end{equation*}

We have thus shown
\begin{lemma} \label{lem:explicit braiding}
	For any unitary representations $D_1$ and $D_2$ of $\caD(G)$ and any good positive half-infinite ribbon $\rho$ we have
	\begin{equation}
		\ep( \chi_{\rho}^{D_1}, \chi_{\rho}^{D_2} ) = B(D_1, D_2).
	\end{equation}
\end{lemma}

The following proposition now follows immediately:

\begin{proposition} \label{prop:braided equivalence of Amp_rho and Rep}
    The functor $F : \Rep_f \caD(G) \rightarrow \Amp_{\rho}$ is an equivalence of braided $\rm C^*$-tensor categories.
\end{proposition}

\begin{proof}
    By Proposition \ref{prop:monoidal equivalence} it suffices to check
    \begin{equation}
	F( B(D_1, D_2) ) = \1 \otimes B(D_1, D_2) = \ep( \chi_{\rho}^{D_1}, \chi_{\rho}^{D_2}  ) = \ep( F(D_1), F(D_2) ).
    \end{equation}
    for any two unitary representations $D_1, D_2$, where we used Lemma \ref{lem:explicit braiding} in the second step.
\end{proof}

\subsection{Equivalence of \texorpdfstring{$\Amp_\rho$}{Amprho} and \texorpdfstring{$\Amp_f$}{Ampf}} \label{subsec:Amp_rho and Amp_f}

Let us first note that $\Amp_f$ is semisimple:
\begin{proposition} \label{prop:amplimorphisms are direct sums of simples}
    Any amplimorphism $\chi \in \Amp_f$ is equivalent to a finite direct sum of irreducible amplimorphisms.
\end{proposition}

\begin{proof}
    This follows immediately from Proposition \ref{prop:existence of subobjects for Amp} and the assumption that all objects of $\Amp_f$ are \emph{finite} amplimorphisms.
\end{proof}

\begin{proposition}
\label{prop:equivalence of Amp_rho and Amp}
    The categories $\Amp_\rho$ and $\Amp_f$ are equivalent as $\rm C^*$-categories. In particular, $\Amp_f$ is closed under the tensor product of $\Amp$, so that $\Amp_f$ is a full braided $\rm C^*$-tensor subcategory of $\Amp$.
\end{proposition}

\begin{proof}
    Recall Proposition \ref{prop:T to t} which shows that $\Amp_{\rho}$ is a full $\rm C^*$-subcategory of $\Amp_f$. Let us consider the functor $F : \Amp_{\rho} \rightarrow \Amp_f$ which embeds $\Amp_{\rho}$ into $\Amp_f$. We want to show that $F$ is an equivalence of $\rm C^*$-categories. Clearly, $F$ is linear, fully faithful, and respects the $*$-structure. The only thing that remains to be shown is that $F$ is essentially surjective, but this follows from Propositions \ref{prop:amplimorphisms are direct sums of simples} and \ref{prop:charcterization of simple amplis}.

    It follows that for any two amplimorphisms $\chi_1$ and $\chi_2$ of $\Amp_f$ there are representations $D_1$ and $D_2$ such that $\chi_1$ is equivalent to $\chi_{\rho}^{D_1}$ and $\chi_2$ is equivalent to $\chi_{\rho}^{D_2}$, and therefore $\chi_1 \times \chi_2$ is equivalent to $\chi_{\rho}^{D_1} \times \chi_{\rho}^{D_2} = \chi_{\rho}^{D_1 \times D_2}$ (see Lemma \ref{lem:sum and product of amplimorphism representations}). In particular, $\chi_1 \times \chi_2$ is finite (Proposition~\ref{prop:T to t}) and so $\Amp_f$ is closed under the tensor product. It is therefore a $\rm C^*$-tensor subcategory of $\Amp$, and inherits the braiding from $\Amp$.
\end{proof}

\begin{proposition}
\label{prop:braided monoidal equivalence of Amp_rho and Amp}
    The categories $\Amp_\rho$ and $\Amp_f$ are equivalent as braided $\rm C^*$-tensor categories.
\end{proposition}

\begin{proof}
    From Proposition \ref{prop:equivalence of Amp_rho and Amp} the embedding functor $F : \Amp_{\rho} \rightarrow \Amp_f$ is an equivalence of $\rm C^*$-categories, and $\Amp_{\rho}$ and $\Amp_f$ are braided $\rm C^*$-tensor subcateogries of $\Amp$. Clearly $F$ is monoidal and braided, which proves the claim.
\end{proof}

\subsection{Proof of Theorem \ref{thm:main result}}
 \label{subsec:proof of Theorem}

Before proving the main theorem, we must first establish that $\DHR_f$ is closed under the tensor product and therefore inherits the braided $\rm C^*$-tensor structure of $\DHR$.
\begin{lemma} \label{lem:DHR_f is braided monoidal}
    The full subcategory $\DHR_f$ of $\DHR$ is closed under the tensor product. It is therefore a braided $\rm C^*$-tensor subcategory of $\DHR$ whith braiding inherited from $\DHR$.
\end{lemma}

\begin{proof}
    Let $\nu_1$ and $\nu_2$ be endomorphisms belonging to $\DHR_f$. By Lemma \ref{lem:ampli is equivalent to endo} there are amplimorphisms $\chi_1$ and $\chi_2$ belonging to $\Amp$ such that $\nu_1$ is equivalent to $\chi_1$ and $\nu_2$ is equivalent to $\chi_2$. Moreover, since $\nu_1$ and $\nu_2$ are finite, so are $\chi_1$ and $\chi_2$. \ie $\chi_1$ and $\chi_2$ belong to $\Amp_f$. It follows that $\nu_1 \times \nu_2$ is equivalent to $\chi_1 \times \chi_2$, which is finite by Proposition \ref{prop:equivalence of Amp_rho and Amp}. This shows that $\nu_1 \times \nu_2$ is finite and so $\DHR_f$ is closed under the tensor product.
\end{proof}

We now proceed to prove our main result, Theorem \ref{thm:main result}, which we restate here for convenience.

\begin{theorem}
    The categories $\Amp_f$, $\DHR_f$, and $\Rep_f \caD(G)$ are all equivalent as braided $\rm C^*$-tensor categories.
\end{theorem}

\begin{proof}
    With Propositions \ref{prop:braided equivalence of Amp_rho and Rep} and \ref{prop:braided monoidal equivalence of Amp_rho and Amp} establishing the equivalence of $\Rep_f \caD(G)$ and $\Amp_f$, all that remains to be shown is the equivalence of $\DHR_f$ and $\Amp_f$ as braided $\rm C^*$-tensor categories.

    To see this, let $F : \DHR_f \rightarrow \Amp_f$ be the embedding functor. Clearly $F$ is linear, fully faithful, braided monoidal, and respects the $*$-structure. It remains to check that $F$ is essentially surjective, but this is immediate from Lemma \ref{lem:ampli is equivalent to endo}.
\end{proof}

\begin{remark}
As mentioned previously, we restrict to the category $\DHR_f$.
Since dualizable DHR endomorphisms are automatically finite (in our sense of the terminology) by~\cite{LongoRoberts97}, and all objects in the category $\Amp_f$ are dualizable, our results imply that the restriction of the category $\mathcal{O}_{\Lambda_0}$ (as defined by Ogata~\cite{Ogata2021}) to dualizable sectors (i.e., those who admit a conjugate) is precisely $\Rep_f \caD(G)$.
We do not expect that $\mathcal{O}_{\Lambda_0}$ has any objects which are not equivalent to (possibly infinite) direct sums of objects in $\Amp_f$.
For example, any simple direct summand of any such an object would be equivalent to a simple object in~$\Amp_f$.
\end{remark}


\section{Conclusions}

We explicitly characterized the category of anyon sectors for Kitaev's quantum double model for finite groups $G$.
As conjectured, the answer is that it is braided monoidally equivalent to $\Rep_f\,\caD(G)$.
This provides the first example where the category of anyon sectors is constructed explicitly for a model with non-abelian anyons.

The problem is tractable for the quantum double model largely because the Hamiltonian is of commuting projector type.
In general, we are interested in the whole quantum \emph{phase}.
The Hamiltonian of the quantum double model has a spectral gap in the thermodynamic limit, and roughly speaking another state is said to be in the same phase as the frustration free ground state $\omega_0$ of the quantum double model if they can be realised as ground states of a continuous path of gapped Hamiltonians.\footnote{Alternatively, it is possible to give a definition of a phase without referring to Hamiltonians at all, using e.g. finite depth quantum circuits or suitable locality preserving automorphisms.}
Using standard techniques (which we outline below) our results carry over to other states in the same gapped phase, which may no longer be ground states of a commuting projector Hamiltonian. 
One of the features of the quantum double model is that the physical features should be stable against small perturbations.
Indeed, the ground state has what is called local topological quantum order (LTQO)~\cite{Fiedler2014,Cui2020kitaevsquantum}.
This implies that sufficiently small local perturbations (even if applied throughout the system) do not close the spectral gap~\cite{michalakisz,brahami:2010}.

The result mentioned above implies that the ground states of the unperturbed and perturbed quantum double models can be related via an automorphism of $\alg{A}$ which is sufficiently local (meaning it satisfies a Lieb--Robinson type bound)~\cite{bachmannmns}.
Hence one can consider the \emph{phase} of a ground state as all states that can be connected via such a sufficiently local automorphism.
It turns out that the braided category of anyon sectors is an invariant of such a phase (that is, each state in the phase supports the same type of anyons).
This follows from the work of Ogata~\cite{Ogata2021} (see also~\cite{Ogata2021b} for a review), applied to the category $\DHR$ (or $\DHR_f$).
Alternatively, one can apply the approximation techniques developed there (necessary because one is forced to replace Haag duality by a weaker, approximate version) directly to the amplimorphisms constructed here. 

\appendix


\section{The quantum double of a finite group and its category of representations}
\label{app:introduction to D(G)}

Fix a finite group $G$. For any $g \in G$ we write $\bar g := g^{-1}$ for its inverse. We denote the unit of $G$ by $1 \in G$.
The quantum double algebra $\caD(G)$ of the finite group $G$ consists of formal $\C$-linear combinations of pairs of group elements $(g, h) \in G \times G$ equipped with product $\mu$, unit $\eta$, coproduct $\Delta$, counit $\ep$, and antipode $S$ defined by the linear extensions of
\begin{align*}
    \mu \big( (g_1, h_1), (g_2, h_2) \big) = \delta_{g_1, h_1 g_2 \bar h_1} (g_1, h_1 h_2), &\quad \Delta( g, h ) = \sum_{k \in G} (k, h) \otimes (\bar k g, h) \\
    \eta(1) = \sum_{k \in G} (k, 1), \quad \ep( g, h ) &= \delta_{g, 1} \quad S(g, h) = (\bar h \bar g h, \bar h),
\end{align*}
giving $\caD(G)$ the structure of a Hopf algebra. It is in fact a Hopf $*$-algebra with $(g, h)^* = (\bar h g h, \bar h)$, and is quasi-triangular with universal $R$-matrix
\begin{equation} \label{eq:R-matrix}
    R = \sum_{g, k \in G } (k, g) \otimes (g, 1).
\end{equation}

Let us recall some basic facts about the representation theory of $\caD(G)$ (see e.g.~\cite{Gould1993}) and establish notation. Denote by $\Rep_f \caD(G)$ the $\rm C^*$-category of finite dimensional unitary representations of $\caD(G)$, \ie representations $D$ such that $D(a^*) = D(a)^*$ for all $a \in \caD(G)$. We let $(D_2 | D_1)$ be the space of intertwiners from a representation $D_1$ to a representation $D_2$. We denote by $I$ the finite set of equivalence classes of irreducible representations and for each $i \in I$ we fix a representative $D^{(i)}$ from $i$. The algebra $\caD(G)$ is semisimple, from which it follows that any representation in $\Rep_f \caD(G)$ is equivalent to a direct sum of fintely many copies of the representatives $\{D^{(i)}\}_{i \in I}$.

The coproduct of $\caD(G)$ gives a monoidal product $\times$ of representations via
\[
(D_1 \times D_2)(a) := (D_1 \otimes D_2)(\Delta(a)),
\]
making $\Rep_f \caD(G)$ into a $\rm C^*$-tensor category. For $i, j \in I$ we have a unitary equivalence
\begin{equation*}
    D^{(i)} \times D^{(j)} \simeq \bigoplus_{k \in I} N_{ij}^k \cdot D^{(k)}
\end{equation*}
where the non-negative integers $N_{ij}^k$ stand for the multiplicity of $D^{(k)}$ in the direct sum.

The universal $R$-matrix of Eq. \eqref{eq:R-matrix} provides a braiding $B : \times \rightarrow \times^{\rm{op}}$ for $\Rep_f \caD(G)$ whose component maps are 
\begin{equation} \label{eq:braiding of Rep defined}
    B(D_1, D_2) := P_{12} \cdot (D_1 \times D_2) (R),
\end{equation}
where $P_{12}$ interchanges the factors in the tensor product of the representation spaces of $D_1$ and $D_2$. This makes $\Rep_f \caD(G)$ into a braided $\rm C^*$-tensor category.


\section{Ribbon operators}  \label{app:ribbon operators}
For the convenience of the reader, we recall the definition of ribbon operators and some of their properties, tailored to the triangular lattice we are using in this paper.
The material in this appendix is well-known, see e.g.~\cite{kitaev2003fault,Bombin2008,Yan2022} for more details.

\subsection{Triangles and ribbons}

Denote by $\Gamma^V, \Gamma^F$ the set of vertices and faces in $\Gamma$ respectively. An oriented edge $e \in \vec{\Gamma}^E$ may be identified with the pair of vertices $e = (v_0, v_1)$ where $v_0$ is the origin and $v_1$ the target of $e$. We write $\partial_0 e = v_0$ and $\partial_1 e = v_1$, and we have $\bar e = (v_1, v_0)$.

We say a vertex $v$ belongs to a face $f$ if $v$ sits on the boundary of $f$. A site $s = (v, f)$ is a pair of a vertex $v$ and a face $f$ such that $v$ belongs to $f$. We write $v(s) = v$ and $f(s) = f$.

Let $\bar\Gamma$ be the dual lattice to $\Gamma$. To each edge $e \in \Gamma^E$ we associate the oriented dual edge $e^*$ which crosses $e$ from right to left as follows : \includegraphics[width=0.6cm]{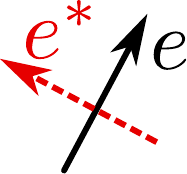}

A direct triangle $\tau = (s_0, s_1, e)$ consists of a pair of sites $s_0, s_1$ that share the same face, and an edge $e \in \Gamma^E$ connecting the vertices of $s_0$ and $s_1$. We write $\partial_0 \tau = s_0$ and $\partial_1 \tau = s_1$ for the initial and final sites of the direct triangle $\tau$, and $e_{\tau} = (v(s_0), v(s_1))$ for the oriented edge associated to $\tau$. The opposite triangle to $\tau$ is the direct triangle $\bar \tau = (s_1, s_0, e)$. Similarly, a dual triangle $\tau = (s_0, s_1, e)$ consists of a pair of sites $s_0, s_1$ that share the same vertex, and the edge $e$ whose associated dual edge $e^*$ connects the faces of $s_0$ and $s_1$. We write again $\partial_0 \tau = s_0$ and $\partial_1 \tau = s_1$, $e^*_{\tau} = (f(s_0), f(s_1))$ for the oriented dual edge associated to $\tau$, and define an opposite dual triangle $\bar \tau = (s_1, s_0, e)$.

A ribbon $\rho = \{ \tau_i \}_{i = 1}^l$ is a finite tuple of triangles such that $\partial_1 \tau_{i} = \partial_0 \tau_{i+1}$ for all $i = 1, \cdots, l-1$ and such that for each edge $e \in \Gamma^E$ there is at most one triangle $\tau_i$ for which $\tau_i = (\partial_0 \tau_i, \partial_1 \tau_i, e)$. We define $\partial_0 \rho = \partial_0 \tau_1$ and $\partial_1 \rho = \partial_1 \tau_l$. If $\rho$ consists of only direct triangles we say that $\rho$ is a direct ribbon, and if $\rho$ consists of only dual triangles we say ther $\rho$ is a dual ribbon. The empty ribbon is denoted by $\ep = \emptyset$.

A ribbon is positively oriented (positive for short) if the sites of all its direct triangles lie to the right of their edges along the orientation of $\rho$ and vice versa for its dual triangles. The ribbon is negatively oriented (negative) otherwise, \cf Figure \ref{fig:positive and negative ribbon}. 

\begin{figure}[!ht]
\centering
\includegraphics[width = 0.6\textwidth]{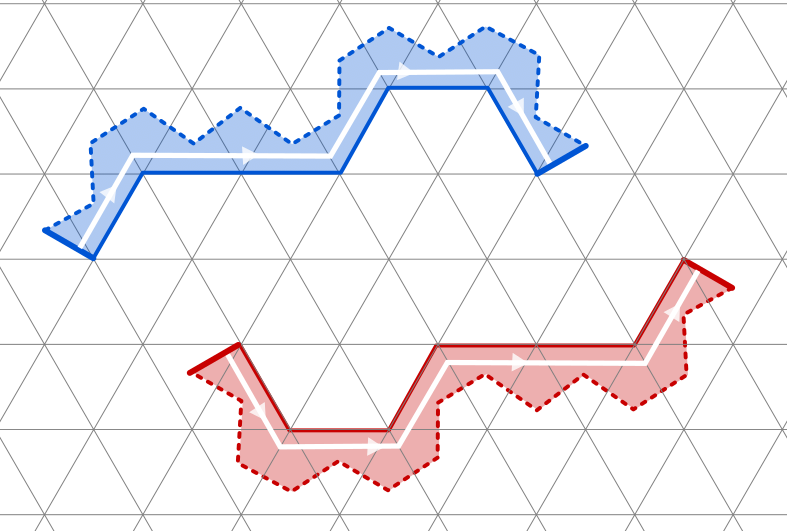}
\caption{An example of a positive ribbon (in red) and a negative ribbon (in blue).}
\label{fig:positive and negative ribbon}
\end{figure}

If we have two ribbons $\rho_1 = \{\tau_i\}_{i = 1}^{l_1}$ and $\rho_2 = \{ \tau_i \}_{i = l_1 + 1}^{l_1 + l_2}$ satisfying $\partial_1 \rho_1 = \partial_0 \rho_2$ then we can concatenate them to obtain a ribbon $\rho_1 \rho_2 = \{ \tau_i \}_{i = 1}^{l_1 + l_2}$. If $\rho_1$ and $\rho_2$ are non-empty then $\partial_0 \rho = \partial_0 \rho_1$ and $\partial_1 \rho = \partial_1 \rho_2$. The opposite ribbon to $\rho = \{\tau_i\}_{i = 1}^l$ is the ribbon $\bar \rho = \bar \tau_l \cdots \bar \tau_1$. If $\rho$ is positively oriented then $\bar \rho$ is negatively oriented and vice versa. The support of a ribbon $\rho = \{ \tau_i = (s_0^{(i)}, s_1^{(i)}, e_i)  \}_{i=1}^l$ is $\supp(\rho) := \{ e_i \}_{i = 1}^l.$

\subsection{Ribbon operators}

\subsubsection{Definitions and basic properties}

For each edge $e \in \Gamma^E$ we define the following operators on $\caH_e$:
\begin{equation}
	L_e^h := \sum_{g \in G} \, |hg \rangle \langle g|, \quad R_e^h := \sum_{g \in G} \, | g \bar h \rangle \langle g |, \quad T_e^g := | g \rangle \langle g|.
\end{equation}
The $L_e^h$ and $R_e^h$ are unitaries, implementing the left and right group action on $\caH_e$ respectively. The $T_e^g$ are projectors.

To each dual triangle $\tau = (s_0, s_1, e)$ we associate unitaries $L^h_{\tau}$ supported on the edge $e$ defined as follows. If $e^* = (f(s_0), f(s_1))$ and $v(s_0) = \partial_0 e$ then $L^h_{\tau} := L_e^h$. If $e^* = (f(s_0), f(s_1))$ and $v(s_0) = \partial_1 e$ then $L^h_{\tau} := R_e^{\bar h}$. If $e^* = (f(s_1), f(s_0))$ and $v(s_0) = \partial_0 e$ then $L^h_{\tau} := L_e^{\bar h}$. Finally, If $e^* = (f(s_1), f(s_0))$ and $v(s_0) = \partial_1 e$ then $L^h_{\tau} := R_e^h$. Similarly, to each direct triangle $\tau = (s_0, s_1, e)$ we associated projectors $T^g_\tau := T^g_{e}$ if $e = (v(s_0), v(s_1))$ and $T^g_{\tau} := T_e^{\bar g}$ if $e = (v(s_1), v(s_0))$.

To each finite ribbon $\rho$ we associate ribbon operators $F^{h, g}$ as follows. If $\ep$ is the trivial ribbon then $F_{\ep}^{h, g} := \delta_{g, 1} \I$. For ribbons composed of a single direct triangle $\tau$ we put $F_{\tau}^{h, g} = T_{\tau}^g$, and for ribbons composed of a single dual triangle $\tau$ we put $F_{\tau}^{h, g} = \delta_{g, 1} L_{\tau}^h$. The ribbon operators for general ribbons are defined inductively by the formula
\begin{equation}
\label{eq:ribbon_decomposition}
	F_{\rho_1 \rho_2}^{h, g} = \sum_{k \in G} \, F_{\rho_1}^{h, k} \, F_{\rho_2}^{\bar k h k, \bar k g}.
\end{equation}
The ribbon operator $F_{\rho}^{h, g}$ is supported on $\supp(\rho)$. Let us define
\begin{equation}
	T_{\rho}^g := F_{\rho}^{e, g}, \quad L_{\rho}^h := \sum_{g \in G} \, F_{\rho}^{h, g}.
\end{equation}
Then $F_{\rho}^{h, g} = L_{\rho}^h T_{\rho}^g = T_{\rho}^g L_{\rho}^h$. The multiplication and adjoint of ribbon operators is given by
\begin{equation} \label{eq:ribbon multiplication and adjoint}
	F_{\rho}^{h_1, g_1} \cdot F_{\rho}^{h_2, g_2} = \delta_{g_1, g_2}  F_{\rho}^{h_1 h_2, g_1}, \quad \big( F_{\rho}^{h, g} \big)^* = F_{\rho}^{\bar h, g},
\end{equation}
and reversing the orientation of a ribbon yields
\begin{equation} \label{eq:ribbon reversal}
	F_{\rho}^{h, g} = F_{\bar \rho}^{\bar g \bar h g, \bar g}.
\end{equation}
Note the natural appearance of the antipode of $\caD(G)$.

We also have the following property:
\begin{equation} \label{eq:sum to identity}
    \sum_k F_\rho^{e,k} = \I.
\end{equation}

If we have two positive ribbons $\rho_1$ and $\rho_2$ with common initial site as in Figure \ref{fig:braiding positive ribbons} then (\cf Eq.~(38) of \cite{kitaev2003fault}):
\begin{equation} \label{eq:braiding positive ribbons}
	F_{\rho_2}^{g_2, h_2} F_{\rho_1}^{g_1, h_1} = F_{\rho_1}^{g_1, h_1} F_{\rho_2}^{\bar g_1 g_2 g_1, \bar g_1 h_2}.
\end{equation}

\subsubsection{Gauge transformations and flux projectors} \label{subsubsec:gauge transformations and flux projections}

For any site $s$ there is a unique counterclockwise closed direct ribbon with end sites equal to $s$ which we denote by $\rho_{\triangle}(s)$. Similarly, there is a unique counterclockwise closed dual ribbon with end sites equal to $s$ which we denote by $\rho_{\star}(s)$. For any site $s$ we define gauge transformations $A_s^h$ and flux projectors $B_s^g$ by
\begin{equation} \label{eq:gauge transformations and flux projectors defined}
	A_s^h := L_{\rho_{\star}(s)}^{h, e}, \quad B_s^g := T_{\rho_{\triangle}(s)}^{e, g}.
\end{equation}

\begin{figure}[!ht]
    \centering
    \includegraphics[width=0.5\linewidth]{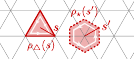}
    \caption{Example of $\rho_{\triangle}(s)$ and $\rho_\star(s')$.}
    \label{fig:enter-label}
\end{figure}

Let us define $U_s : \caD(G) \rightarrow \caA$ by
\begin{equation}
	U_s(g, h) := B_s^g A_s^h,
\end{equation}
extended linearly to $\caD(G)$. One easily checks that $U_s$ is an injective homomorphism of *-algebras, \ie the $B_s^{g} A_s^h$ span a representation of the quantum double algebra $\caD(G)$.

We note that the gauge transformations $A_s^h$ depend only on the vertex $v = v(s)$, so we may write $A_v^h := A_{s}^h$ for any site $s$ with $v = v(s)$. Moreover, the trivial flux projectors $B_s^e$ depend only on the face $f = f(s)$ so we write $B_f := B_s^e$ for any site $s$ with $f = f(s)$.

Finally, we define the projector onto states that are gauge invariant at the vertex $v$ by
\begin{equation}
	A_v := \frac{1}{\abs{G}} \sum_{h \in G} \, A_v^h.
\end{equation}
A straightforward calculation shows that this indeed is a projection.

\section{Convergence of transporters}

In this appendix we prove some technical lemmas needed to construct charge transporters.
The following Lemma, which we prove here for convenience, is well-known (c.f.~\cite[Prop. II.4.9]{TakesakiI}).

\begin{lemma} \label{lem:abstract unitarity}
    Let $\alg{A} \subset \mathfrak{B}(\mathcal{H})$ be a $*$-algebra acting on some Hilbert space $\mathcal{H}$.
    Suppose that $\mathcal{H}_0 \subset \mathcal{H}$ is a dense subset of vectors.
    Let $U_\lambda \in \alg{A}$ be a uniformly bounded net such that for each $\xi \in \mathcal{H}_0$ both $U_\lambda \xi$ and $U_\lambda^* \xi$ converge in the norm topology of $\mathcal{H}$.
    Then $U_\lambda$ converges to some $U \in \alg{A}''$ in the strong-$*$ operator topology.
    If moreover each $U_\lambda$ is unitary, then the limit $U$ is unitary as well.
\end{lemma}

\begin{proof}
    Choose $\epsilon > 0$ and $\xi \in \mathcal{H}$.
    Then there is $\xi_0 \in \mathcal{H}_0$ such that $\| \xi - \xi_0 \| < \epsilon$.
    By assumption, there is $M > 0$ such that $\| U_\lambda \| < M$ for all $\lambda$.
    From this we get
    \[
        \| (U_\lambda - U_\mu) \xi \| = \| (U_\lambda - U_\mu) (\xi-\xi_0)  + (U_\lambda-U_\mu) \xi_0\| \leq 2 M \epsilon + \| (U_\lambda-U_\mu) \xi_0 \|.
    \]
    Since $U_\lambda \xi_0$ converges by assumption, it follows that $U_\lambda \xi$ is Cauchy. We can therefore define $U \xi := \lim_\lambda U_\lambda \xi$.
    From the construction it is clear that $U$ is linear, and because $\| U_\lambda \|$ is uniformly bounded, it follows that $U$ is a bounded operator. A similar argument gives us an operator $\tilde{U}^*$, defined via $\tilde{U}^* \xi = \lim_\lambda U_\lambda^* \xi$.
    
    For all $\xi,\eta \in \mathcal{H}_0$ we have
    \[
    \begin{split}
        \left| \langle \eta, (U^* - \tilde{U}^*) \xi \rangle \right| 
        &= 
        \left| \langle \eta, (U^*-U_\lambda^*) \xi \rangle + \langle \eta, (U_\lambda^*-\tilde{U}^*) \xi \rangle \right| \\
        &\leq 
        \| (U-U_\lambda) \eta \| \|\xi\| + \| \eta\| \| (U_\lambda^*-\tilde{U}^*) \xi \|.
    \end{split}
    \]
    Since the right hand side tends to zero, it follows that $\tilde{U}^* = U^*$, and hence strong convergence of $U_\lambda \to U$ and $U_\lambda^* \to U^*$ gives that $U_\lambda \to U$ in the strong-$*$ operator topology.
    Since the ball of radius $M$ in $\alg{A}''$ is closed in the strong-$*$ operator topology, it follows that $u \in \alg{A}''$.
    
    Finally, suppose that the $U_\lambda$ are unitary.
    By strong-$*$ operator convergence, we have
    \[
        \left\| U \xi \right\| = \lim_\lambda \| U_\lambda \xi \| = \|\xi\|, \quad\quad\quad\quad  \| U^* \xi \| = \lim_\lambda \| U_\lambda^* \xi \| = \|\xi\|
    \]
    for all $\xi \in \mathcal{H}$.
    Hence both $U$ and $U^*$ are isometries, and it follows that $U$ is unitary.
\end{proof}
Note that we need to assume that \emph{both} $U_\lambda \xi$ and $U_\lambda^* \xi$ converge. Since the adjoint is not continuous with respect to the strong operator topology, one does not follow from the other.

\begin{lemma} \label{lem:unitary transporters}
    Let $\rho_1$ be a half-infinite positive ribbon starting at the site $s_0$ and $\rho_2$ a half-infinite negative ribbon starting in $s_1$.
    Suppose that $\{\xi_n\}_{n \in \mathbb{N}}$ is a bridge from $\rho_1$ to $\rho_2$ in the sense of Definition~\ref{def:bridge}, and write $\sigma_n = \rho_{1,n} \xi_n \rho_{2,n}$ as in that definition.
    Finally, choose $g,h \in G$.
    Then $\pi_0(F_{\sigma_n}^{h,g})$ converges in the strong-* operator topology to an operator $F \in \pi_0(\alg{A})''$.
\end{lemma}

\begin{proof}
    Recall that $(\pi_0, \mathcal{H}, \Omega)$ is the GNS representation for the frustration free ground state $\omega_0$ of the quantum double model. To ease the notation we omit $\pi_0$ on the operators.

    Let $A \in \caA^{\loc}$.
    Then there is some $k \in \mathbb{N}$ such that $\supp(A)^{+1} \cap \sigma_n$ is constant for all $n \geq k$, where the $+1$ superscript denotes a ``fattening'' of the set $\supp(A)$ by one site.
    For all $n \geq k$, write $\rho_{i,n\setminus k}$ for the (finite) ribbon $\rho_{i,n}$ with the first $k$ triangles removed, and define $\widehat{\xi}_n = \rho_{1,n \setminus k} \xi_n \rho_{2,n \setminus k}$.
    That is, $\sigma_n = \rho_{1,k} \widehat{\xi}_n \rho_{2,k}$.
    It follows from the choice of $k$ that $\supp(A)^{+1} \cap \widehat{\xi}_n = \emptyset$ for all $n \geq k$.
    Then, using the decomposition rule for ribbon operators, Eq.~\eqref{eq:ribbon_decomposition}, we have for all $n \geq k$ that
    \[
    \begin{split}
        F_{\sigma_n}^{h,g} A \Omega 
        &=
        \sum_{m_1, m_2 \in G} F_{\rho_{1,k}}^{h,m_1} F_{\widehat{\xi}_n}^{\bar m_1 h m_1, m_2}  F_{\rho_{2,k}}^{\overline{m_1 m_2} h (m_1 m_2), \overline{m_1 m_2} g} A \Omega \\
        &=
        \sum_{m_1, m_2 \in G} F_{\rho_{1,k}}^{h,m_1}   F_{\rho_{2,k}}^{\overline{m_1 m_2} h (m_1 m_2), \overline{m_1 m_2} g} A F_{\widehat{\xi}_n}^{ \bar m_1 h m_1, m_2} \Omega.
        \end{split}
    \]
    In the last step we used locality of the operators.
    Note that for $n,m \geq k$, the ribbons $\widehat{\xi}_n$ and $\widehat{\xi}_m$ have the same start and end points by construction.
    Since the action of ribbon operators on the ground state depends only on the endpoints of the ribbons (see e.g.~\cite{Bombin2008,HamdanThesis}) we have that $F_{\widehat{\xi_n}}^{\bar m_1 h m_1, m_2} \Omega = F_{\widehat{\xi_m}}^{\bar m_1 h m_1, m_2} \Omega$.
    It follows that the sequence $F_{\sigma_n}^{h,g} A \Omega$ converges in norm.
    Because the adjoint of a ribbon operator is again a ribbon operator (on the same ribbon, \cf \eqref{eq:ribbon multiplication and adjoint}), the argument above shows that $(F_{\sigma_n}^{h,g})^* A \Omega$ also converges in norm as $n \to \infty$.
    Note that for ribbon operators we have that $\| F_{\sigma_n}^{h,g} \| \leq 1$, regardless of $\sigma_n$.
    Hence by Lemma~\ref{lem:abstract unitarity}, the result follows.
\end{proof}


\setcounter{biburlnumpenalty}{9000}
\printbibliography

\end{document}